\documentclass[a4paper]{article}
\usepackage{xr-hyper}

\usepackage{amsmath}
\usepackage{amsthm}
\usepackage{amssymb}
\usepackage{amsfonts}
\usepackage{wasysym} 
\usepackage{cancel}
\usepackage{dsfont} 
\usepackage{subfiles}
\usepackage{mathrsfs}
\usepackage{mathtools}
\usepackage{graphics}
\usepackage[utf8]{inputenc}

\usepackage{soul} 

\usepackage[a4paper,bottom=3.0cm,top=2.5cm,left=3.2cm,right=3.2cm]{geometry}

\usepackage{scalerel,stackengine}
\stackMath
\newcommand\reallywidecheck[1]{%
\savestack{\tmpbox}{\stretchto{%
\scaleto{%
\scalerel*[\widthof{\ensuremath{#1}}]{\kern-.6pt\bigwedge\kern-.6pt}%
{\rule[-\textheight/2]{1ex}{\textheight}}
}{\textheight}%
}{0.5ex}}%
\stackon[1pt]{#1}{\scalebox{-1}{\tmpbox}}%
}

\makeatletter
\DeclareRobustCommand\widecheck[1]{{\mathpalette\@widecheck{#1}}}
\def\@widecheck#1#2{%
\setbox\z@\hbox{\m@th$#1#2$}%
\setbox\tw@\hbox{\m@th$#1%
\widehat{%
\vrule\@width\z@\@height\ht\z@
\vrule\@height\z@\@width\wd\z@}$}%
\dp\tw@-\ht\z@
\@tempdima\ht\z@ \advance\@tempdima2\ht\tw@ \divide\@tempdima\thr@@
\setbox\tw@\hbox{%
\raise\@tempdima\hbox{\scalebox{1}[-1]{\lower\@tempdima\box
\tw@}}}%
{\ooalign{\box\tw@ \cr \box\z@}}}
\makeatother

\newtheoremstyle{mystyle}
  {\topsep}{\topsep}{}{}{\bfseries}{ }{\newline}
  {(\thmnumber{#2}) \thmname{#1}\quad\thmnote{(#3)}%
     \ifstrempty{#3}%
      {\addcontentsline{toc}{subsection}{#1~\thethm}}%
      {\addcontentsline{toc}{subsection}{#1~\thethm~(#3)}}%
  }
\newtheoremstyle{mystyle}
  {\topsep}{\topsep}{}{}{\bfseries}{ }{\newline}
  {\(\mathllap{\mbox{(\thmnumber{#2})\;}}\)\thmname{#1}\quad\thmnote{(#3)}%
     \ifstrempty{#3}%
      {\addcontentsline{toc}{subsection}{#1~\thethm}}%
      {\addcontentsline{toc}{subsection}{#1~\thethm~(#3)}}%
    }

\makeatletter

\usepackage{alltt}

\def\VV#1\par{\sloppy \begin{alltt}#1\end{alltt}}
\def\Vr#1\par{{\color{red}\begin{alltt}#1\end{alltt}}}
\def\Vg#1\par{{\color{OliveGreen}\begin{alltt}#1\end{alltt}}}
\def\Vb#1\par{{\color{blue}\begin{alltt}#1\end{alltt}}}

\newcommand{\fl}[1][3em]{\hspace{#1}&\hspace{-#1}} 

\usepackage{hyphenat}
\usepackage{xparse}
\usepackage{microtype}
\DeclareDocumentCommand{\vv}{O{NOTE}O{black}}{{\setlength\emergencystretch{\hsize}\hbadness=10000\color{#2}\noindent\tt[#1]}}%

\newcommand{\vb}[1][NOTE]{\vv[#1][black!40!blue]}

\newcommand{\vbo}[1][NOTE]{\vb[#1]}

\newcommand{\hideopt}{\renewcommand{\vbo}[1][NOTE]{}}

\newcommand\listdefinitionname{List of Definitions}
\newcommand\listofdefinitions{%
\section*{\listdefinitionname}\@starttoc{def}}
\makeatother

\newcommand{\Ii}{\mathrm i} 
\newcommand{\Ee}{\mathrm e} 
\newcommand{\Alg}{\mathfrak A}

\newcommand{\term}[1]{{\it #1}\/}

\newcommand{\Reg}{\mathcal O}

\newcommand{\DoubleCone}{{\mathscr C}}

\newcommand{\PoincareGroup}{\mathcal P}
\newcommand{\LorentzGroup}{\mathcal L}

\newcommand{\RealNum}{\mathbb R}
\newcommand{\ComplexNum}{\mathbb C}

\newcommand{\NaturalNum}{\mathbb N}
\newcommand{\SchwartzSpace}{\mathscr S}
\newcommand{\DInt}[2][\,]{{\mathrm d}^{\hspace{-0.20ex}#1}\hspace{-0.25ex}#2}

\newcommand{\norm}[1]{\left \|#1 \right\|}
\newcommand{\normm}[1]{\|#1\|}
\newcommand{\abs}[1]{\left |#1 \right|}

\newcommand{\CStar}{C^*}

\newcommand{\HilbertSpace}{\mathscr H}
\newcommand{\BoundedOps}{\mathcal B}

\newcommand{\LSpace}{L}

\newcommand{\Id}{\mathds 1}

\newcommand{\Domain}[1]{{\mathrm D}_{#1}}

\newcommand{\FiniteEnergyStates}{\mathcal D}

\newcommand{\ContinuousFuncs}{\mathnormal C}

\newcommand{\IndicatorFunction}{\mathds 1}

\newcommand{\CharFct}{\IndicatorFunction}

\ifx\BEAMER\true
\usepackage{color}
\usepackage[pdfborderstyle={/S/U/W 0},colorlinks=true,linkcolor=blue,pdfnewwindow=true]{hyperref} 
\usepackage{cleveref}

\hypersetup{linkcolor=blue!30!black!80,citecolor=blue!30!black!80,urlcolor=blue!30!black!80}

\usepackage{tikz}

\usepackage{etoolbox}

\usepackage{enumerate}
\fi

\newcommand{\oldnote}[1]{} 

\newcommand{\uvec}[1]{\underline{#1}}
\renewcommand{\vec}[1]{{\bf #1}} 

\newcommand{\closure}[1]{\overline{#1}}

\DeclareMathOperator*{\support}{supp}

\newcommand\numberthis{\addtocounter{equation}{1}\tag{\theequation}}

\ifx\labelAllEquations\true
\else
\makeatletter 
\def\[{\begin{equation}}%
\def\]{\end{equation}}%
\renewenvironment{align*}{\begin{equation}\begin{aligned}}{\end{aligned}\end{equation}\ignorespacesafterend}%
\makeatother
\fi

\usepackage{framed}

\renewenvironment{leftbar}[1][\hsize]
{%
    \def\FrameCommand
    {%
        {\hspace{-6pt}\color{red}\vrule width 0.5pt}%
        \hspace{5.5pt}
    }%
    \MakeFramed{\hsize#1\advance\hsize-\width\FrameRestore}%
}
{\endMakeFramed}

\DeclareDocumentCommand{\beginleftbarcw}{ O{blue} O{0.5pt} O{\hsize} }
{%
    \def\FrameCommand
    {%
      {\hspace{-5pt}\hspace{-#2}\color{#1}\vrule width #2}%
        \hspace{5pt}
    }%
    \MakeFramed{\hsize#3\advance\hsize-\width\FrameRestore}%
}

\let\endleftbarcw\endMakeFramed

\renewenvironment{leftbar}[1][\hsize]
{%
  \beginleftbarcw[red][0.5pt][#1]
}
{\endleftbarcw}

\newcommand{\standardThms}{\newcounter{thm}

\newtheorem{Prop}[thm]{Proposition}

\newtheorem{Lem}[thm]{Lemma}
\newtheorem{Cor}[thm]{Corollary}
\newtheorem*{Cor*}{Corollary}
\newtheorem{Thm}[thm]{Theorem}
\newtheorem*{Thm*}{Theorem}

\newtheorem{Def}[thm]{Definition}

\theoremstyle{remark}
\newtheorem{Rem}[thm]{Remark}

\newcommand{\solution}{\proof[Solution.]}
}

\usepackage[
  backend=biber,
  natbib=true,
  style=alphabetic,
  sorting=nyt,
  doi=true,
  isbn=false,
  url=false,
  backref=true,
  giveninits=true 
]{biblatex}

\AtBeginBibliography{
}

\AtEveryCitekey{\clearfield{shorthand}}
\renewbibmacro{in:}{}

\hypersetup{pdfnewwindow}
\usepackage{pdfpages}
\makeatletter
\let\obib\@bibitem
\let\obibl\@lbibitem
\makeatother
\makeatletter
\let\@bibitem\obib
\let\@lbibitem\obibl
\makeatother

\def\WO{\mathbf W}

\DeclareMathOperator{\Span}{span}

\newcommand{\FockSpace}{\mathscr F}

\newcommand{\SymmetricGroup}{\mathfrak S}

\def\Pp{{\bf P}}

\standardThms

\def\Ww{\mathcal W} 
\def\WO{\mathbb W} 
\def\WwR{{\Ww_{\mathrm R}}} 
\def\WwL{{\Ww_{\mathrm L}}} 
\def\Wwr{\WwR} 
\def\T{\tau}
\def\Pp{\boldsymbol P}
\def\s{s}

\newcommand{\HCol}[1]{{#1}}
\newcommand{\Aa}{\HCol A}
\newcommand{\Bb}{\HCol B}

\newcommand{\BB}{\HCol{\mathcal B}}

\crefformat{footnote}{#2\footnotemark[#1]#3}

\hideopt 
\def\ifhide{\iftrue}

\def\itref#1{(\ref{#1})}

\usepackage{tocloft} 
\def\FockSpace{\Gamma}
\def\VelocitySupport{{\mathcal V}}
\def\FS{\FockSpace}
\def\VS{\VelocitySupport}

\def\C{{\mathrm c}}

\usepackage{amsopn}
\def\fin{{\mathrm f}}
\def\ini{{\mathrm i}}
\def\ifintro{\iffalse}
\def\ifintro{\iftrue} 

\def\BB{B}

\usepackage{cancel}
\usepackage{color}
\newcommand{\mail}[1]{\href{mailto:#1}{#1}}

\newcommand\hidden[1]{}
\title{
  Asymptotic Completeness
  in a Class of Massive Wedge-Local Quantum Field Theories
  in any Dimension
  }

  \author{Maximilian Duell$^1$, Wojciech Dybalski$^2$}
 \date{$^1$Mathematisches Institut\\Ludwig-Maximilians-Universität
    München\\
    Theresienstr.\ 39,
    80333 München\\
    {\small E-mail: \mail{duell@math.lmu.de}}\\[0.5em]
    $^2$Faculty of Mathematics and Computer Science\\ 
  Adam Mickiewicz University in Pozna\'n\\ 
  ul.\ Uniwersytetu Pozna\'nskiego 4, 61--614 Pozna\'n \\ {\small E-mail: \mail{wojciech.dybalski@amu.edu.pl}}}

\usepackage{soul}

\def\ifextapp{\iftrue}
\def\iftodo{\iffalse}
\def\iftodo{\iftrue}
\usepackage{soul} 
\newcommand{\AlgR}{\Alg^{0r}}

\newcommand{\sPone}{d}

\def\FEnV{\FiniteEnergyStates}
\def\Cinfty{{\mathcal C}^\infty}

\DeclareFieldFormat{eprint:mp-arc}{%
  mp-arc\addcolon\space
  \ifhyperref
    {\href{http://web.ma.utexas.edu/mp_arc-bin/mpa?yn=#1}{%
       \nolinkurl{#1}%
       \iffieldundef{eprintclass}
         {}
         {\addspace\texttt{\mkbibbrackets{\thefield{eprintclass}}}}}}
    {\nolinkurl{#1}%
     \iffieldundef{eprintclass}
       {}
       {\addspace\texttt{\mkbibbrackets{\thefield{eprintclass}}}}}}
\begin{document}

\maketitle
\begin{abstract}
  A recently developed \(n\)-particle scattering theory for
  wedge-local quantum field theories  is applied
  to a class of models described and constructed by Grosse, Lechner, Buchholz, and Summers.
  In the BLS-deformation setting we establish explicit expressions for \(n\)-particle wave
  operators and the \(S\)-matrix of ordered asymptotic states, and we show
  that ordered asymptotic completeness is stable under the general
  BLS-deformation construction.  In particular, the (ordered) Grosse-Lechner
  \(S\)-matrices are non-trivial also beyond two-particle scattering
  and factorize into 2-particle scattering processes, which
  is an unusual feature in space-time dimension \(d > 1+1\).
  Most notably, the Grosse-Lechner models provide the first examples
  of relativistic (wedge-local) QFT in space-time dimension~\(d > 1+1\) which are
  interacting and asymptotically complete.
\end{abstract}
\tableofcontents

\section{Introduction}
{
\setlength{\parskip}{0.5em}
Asymptotic completeness (AC) of a quantum field theory  (QFT) is the requirement that incoming and outgoing
multi-particle scattering states~\(\Psi^\pm\) should be dense in the Hilbert space
\(\HilbertSpace\) of a given QFT model.
This conceptual question was already raised early on during the development of
the quantum theory of fields \cite{Gre61,Ru62}. But  even after many decades of
research, mathematical results on AC of interacting QFT
models have remained rather scarce (e.g.\ \cite{CD82,DG99,IM06,Le06}) and
there are various obstacles of physical and technical nature (cf.\ \cite{DG13}).
In fact, the first complete proof of asymptotic completeness of an interacting
QFT has been established rather recently: It was obtained by Lechner \cite{Le06t,Le06}
for certain two-dimensional models with factorizing \(S\)-matrix, such as the Sinh-Gordon model.

In the present work we provide, to the best of our knowledge, the first full proof of
AC in a class of interacting wedge-local QFT models on Minkowski space-times of
dimension \(d > 1+1\). The models in question, found by Grosse, Lechner,
Buchholz and Summers in \cite{GL07,BS08, BLS10}, may not be local, but  have the
weaker property of wedge locality.\footnote{Asymptotically complete fully
non-local models have been constructed, e.g.\ \cite{BW84}.} The constructive
procedure developed in \cite{BLS10}, which we will call a BLS-deformation,
preserves the general structure of a relativistic wedge-local theory and in
addition introduces interaction. For example, the BLS-deformation of a free
field  theory, which we will call a GL-model, has a non-trivial two-particle
scattering matrix \cite{GL07,BS08}.  However, collision processes involving \(n
\geq 3\)   particles were not investigated in these works, since $n$-particle
scattering theory was not available in the wedge-local setting back then. Such a
scattering theory  has been developed meanwhile in \cite{MD17a}, so that the
question of asymptotic completeness can now be posed and positively answered for
these models.  More than that, by adapting and generalizing methods from
\cite{BLS10,DT10} we will establish explicit expressions for the general effect
of BLS-deformations on asymptotic states and scattering data. In this manner we
also prove stability of asymptotic completeness under BLS-deformations for
general wedge-local models with massive particles.

To construct \(n\)-particle scattering states according to \cite{MD17a},
we define Haag-Ruelle operators~\(\Bb_{k\T}(f_k)\), \(\T \in \RealNum\), \(1 \leq k \leq n\),
which create the desired one-particle states from the vacuum~\(\Omega\in
\HilbertSpace\).  Then  outgoing
and incoming velocity-ordered scattering states are
\begin{align*}
  \Psi^\pm := \lim_{\T\to\pm\infty} \Bb_{1\T}(f_1) \ldots \Bb_{n\T}(f_n) \Omega.
  \label{eq:introscatt-zero} \numberthis
\end{align*}
The momentum-space configuration of such scattering states is specified
by a family of regular Klein-Gordon solutions~\(f_k\).
For {\em wedge-ordered velocity configurations}\/ (as defined below in
\eqref{eq:precursor-ordering}), the wedge-local
Haag-Ruelle theorem establishes convergence in (\ref{eq:introscatt-zero}) and
Fock structure of the resulting asymptotic states~\cite{MD17a}.
This result holds in general wedge-local theories with massive
particle spectrum, given that all~\(\Bb_{k\T}(f_k)\) are constructed from a common
localization wedge~\(\Ww\).
For the concrete case of the standard wedge~\(\Ww
  = \WwR = \{ (t, \vec x) \in \RealNum^d : \abs{t}<x^1\}\), the velocity ordering
  of outgoing states (\(\T \to + \infty\)) reads
\[
  \VS_{f_n} \prec_{\WwR}
  \VS_{f_{n-1}} \prec_{\WwR}  \ldots \prec_{\WwR }\VS_{f_{1}}. \label{eq:precursor-ordering}\numberthis
\]
Here \(\VS_{f_k} := \{ (1,\vec v_m(\vec p)) \in \RealNum^d : \vec p \in \support \tilde f_k \}\)
denote the supports of the Klein-Gordon wave packets
\[f_k(t, \vec x) := \int \frac{\DInt[d-1] k}{(2\pi)^{d-1}}
  \; \Ee^{-\Ii \omega_m(\vec k) t + \Ii \vec k \cdot \vec x}
  \tilde f_k(\vec k) , \quad \omega_m(\vec k) := \sqrt{\vec k^2+m^2},
  \numberthis
\]
with respect to velocity~\( \vec v_m(\vec p) := \vec p/\omega_m(\vec p) \).  Further, the
{\em precursor}\/ ordering relation in \eqref{eq:precursor-ordering} is defined by \(\VS_{f_k} \prec_{\WwR} \VS_{f_{k-1}} :\Longleftrightarrow
 \VS_{f_{k-1}} - \VS_{f_k}\subset \WwR \) and corresponds to the geometrical requirement that all
relative velocities \(\vec v_{k-1} - \vec v_k \in \VS_{f_{k-1}} - \VS_{f_k}\) yield a
positive directed separation of wave packets~\(f_k\) for \(t \to +\infty\) with respect to the
space-like opening direction of~\(\WwR\), i.e.\ \(
(\vec v_{k-1} - \vec v_{k})^1 > 0\) for all \(2 \leq k \leq n\).
In the present work we prove, in particular, for GL-models that the
velocity-ordered scattering states constructed according to
\eqref{eq:introscatt-zero} and \eqref{eq:precursor-ordering} for each wedge~\(\Ww\) span dense sets~\(\HilbertSpace^{\pm}_\Ww\) in the full Hilbert space
\(\HilbertSpace\) of the interacting model.
In this case, we will say
that a wQFT model has the property of {\em (ordered) asymptotic completeness}.

We now briefly explain the basic ideas of the present analysis of deformed
scattering data and asymptotic completeness in non-technical terms.
In the deformation construction from \cite{BS08,BLS10} the new deformed model is generated by
observables constructed as formal spectral integrals
\[
  A_Q := \int \DInt E_{(H, \Pp)} (p) \; \alpha_{Qp}(A) =  \int  \alpha_{Qp} (A) \;
  \DInt E_{(H, \Pp)}(p)=:{_Q{\!}A} .
  \numberthis
  \label{eq:warpintro}
\]
In the terminology of \cite{BS08,BLS10}, the (smooth) wedge-local operators $A$
of the initial model are {\em `warped'}
with respect to the spectral measure \(\DInt E =\DInt E_{(H, \Pp)}\) of the
energy-momentum operators \(P= (H, \Pp)\). Presently, \(Q \in \RealNum^{d^2}\)
denotes a fixed parameter matrix satisfying certain geometrical
properties~\cite{GL07}. Later it will be chosen as a mapping depending on the localization wedge of the
operator \(A\), as prescribed by Grosse and Lechner, to yield
wedge-locality and Poincaré covariance (if applicable) of the deformed model.

Proceeding now on the formal level of \eqref{eq:warpintro}, let us apply
the wedge-local \(n\)-particle scattering theory from
\cite{MD17a} and compare the scattering states
\begin{align*}
  \Psi^+_0 &:= \lim_{\T\to\infty} \Bb_{1\T}(f_1) \ldots \Bb_{n\T}(f_n) \Omega,
  \label{eq:introscatt}
  \numberthis
  \\
  \Psi^+_Q &:= \lim_{\T\to\infty} \Bb_{1Q\T}(f_1) \ldots \Bb_{nQ\T}(f_n) \Omega,
  \label{eq:introwarpedscatt}
  \numberthis
\end{align*}
where \(\Bb_{kQ\T}(f_k)\), \(1\leq k \leq n\), obtained from \eqref{eq:warpintro} are a
corresponding family of creation operators associated to the deformed model.
Inserting the definition of the warped convolutions appearing in
\eqref{eq:introwarpedscatt} we see directly that
\begin{align*}
  \Psi^+_Q &:= \lim_{\T\to\pm\infty}
  \int \DInt E(p_1) \alpha_{Qp_1}(\Bb_{1\T}(f_1))
  \ldots  \int \DInt E(p_n) \alpha_{Qp_n}( \Bb_{n\T}(f_n)) \Omega.
  \label{eq:psiqdef}
  \numberthis
\end{align*}
On a heuristic level we can thus write \(\Psi^+_Q\)
as a superposition of scattering states constructed in the undeformed model
\[
  \Psi^+_{0; p_1, \ldots, p_n} :=
  \lim_{\T \to \infty} \alpha_{Qp_1}(\Bb_{1\T}(f_1))
  \alpha_{Qp_2}(\Bb_{2\T}(f_2)) \ldots  \alpha_{Qp_n}( \Bb_{n\T}(f_n)) \Omega
  \label{eq:psi0trans}
  \numberthis
\]
which can be related to scattering states of the usual form
\eqref{eq:introscatt}, with operators modified by space-time translations
depending on  \( p_1, \ldots, p_n \in \RealNum^d\). As the action of space-time translations
on the one-particle states can be made explicit, namely \(\alpha_y(\Bb_\T(f)) \Omega =
\Bb_\T(f^y) \Omega\), where \(f^y(x) = f(x-y)\), we may expect on
heuristic grounds that, also for a general
deformed wedge-local model, the corresponding subspaces of scattering states of
the deformed and undeformed model will coincide, that is,
\[\HilbertSpace^+_{Q, \prec_\Ww} = \HilbertSpace^+_{0, \prec_\Ww}
  \subset \HilbertSpace.
  \label{eq:comp}
  \numberthis
\]
Here \(\HilbertSpace\) denotes the full Hilbert space of the interacting model,
and similarly \(\HilbertSpace^-_{Q, \prec_\Ww} = \HilbertSpace^-_{0, \prec_\Ww} \).

This shows that asymptotic completeness of ordered scattering states is stable under
BLS-deformations. Asymptotic completeness of the Grosse-Lechner models follows
from AC of the free massive scalar field theory with respect to ordered
scattering states. Indeed, using that for local QFTs the space of wedge-ordered
scattering states~\(\HilbertSpace^+_{0, \prec_\Ww}\) coincides with the full space of scattering
states~\( \HilbertSpace^+_{0} \),
and exploiting AC of the free field, \(\HilbertSpace^+_{0}
= \HilbertSpace_0\), we obtain for the GL-model
\[
 \HilbertSpace^\pm_{Q, \prec_\Ww} =
 \HilbertSpace^\pm_{0, \prec_\Ww} = \HilbertSpace.
 \label{eq:acord} \numberthis
\]
Hence these wedge-local QFT models are asymptotically complete.

It is the main contribution of the present paper to make the above heuristic
arguments leading from \eqref{eq:psiqdef} to \eqref{eq:comp} mathematically
rigorous. Our analysis is based on oscillatory integral techniques from \cite{BLS10} and the
recent wedge-local \(n\)-particle scattering theory~\cite{MD17a}.
We also profit from the previous analysis of a scattering in
BLS-deformed massless two-dimensional models by Dybalski and Tanimoto
\cite{DT10}.
This case is simplified by the absence of dispersion, which reduces the problem
to two-body scattering of left- and right-movers. Hence wedge-local scattering theory
with \(n\!\leq\!2\) particles suffices. This has in fact been developed and studied much
earlier~\cite{BBS01,Le03,GL07,BS08} via direct application of standard
Haag-Ruelle methods. In particular, the swapping strategy introduced in \cite{MD17a} to
establish scattering theory for arbitrary particle numbers is not needed for
this specific class of \(1\!+\!1\)-dimensional models.

To conclude this introduction we should remark that physical intuition familiar
from local QFT makes \eqref{eq:acord} very plausible on one hand.
On the other, such intuition may fail in the much larger class of general wedge-local models.
Interesting examples can already be found among recently constructed
free product models \cite[Sec.\ 5]{LTU17}.  There it was shown that ordered
two-particle scattering states are not sufficient for two-particle asymptotic
completeness. In contrast to a previous work on clustering properties
\cite{Sol14}, where scattering theory is also mentioned briefly,  our results do
not assume locality of the underlying wQFT model and in this regard the results
of the present work are much more widely applicable within the general wedge-local
framework.

\paragraph*{Acknowledgments.}

The authors gratefully acknowledge funding by the Deutsche Forschungsgemeinschaft within grants DY107/2-1, DY107/2-2. W.D. was also partially supported by the NCN grant ‘Sonata Bis’ 2019/34/E/ST1/00053.

\section{Preliminaries on Wedge-local QFT (wQFT) and Scattering Theory}

\subsection{Operator-algebraic Framework for Wedge-local QFT}
\label{sec:prelwqft}
We will work in an operator-algebraic setting of wedge-local
quantum field theory on Minkowski space-time~\(\RealNum^{d}\), and our results
are valid in arbitrary spatial dimension~\(s:= d-1\).
The family of wedge regions is defined as the
orbit~\(\PoincareGroup \Wwr:=\{\lambda \Wwr = \Lambda \Wwr + x, \lambda  =
(x,\Lambda) \in \PoincareGroup \} \) of the conventional \term{Rindler
wedge}~\(\Wwr := \{ (t, \vec x) \in
\RealNum^{d}: \abs{t} < x^1 \} \) (also, \term{standard wedge} or \term{right wedge})  under the action of the Poincaré
  group~\(\PoincareGroup = \RealNum^d \rtimes \LorentzGroup \).

A wedge-local quantum field theory model is specified by mathematical objects~\((\Alg, \alpha, \HilbertSpace,
\Omega)\), where \(\HilbertSpace\) is the Hilbert space of pure states
containing the \term{vacuum} as a distinguished unit vector \(\Omega \in
\HilbertSpace\).
The wedge-localization of observables  is described by a
family of von Neumann algebras~\(\Alg(\Ww) \subset
\BoundedOps(\HilbertSpace)\) associated to wedge
regions~\(\Ww\).
Poincaré symmetry acts on the wedge-local algebras~\(\Alg(\Ww)\) by a given group of
isomorphisms \(\alpha_\lambda\) and we denote by
\(\lambda = (x,\Lambda) \in
  \PoincareGroup_+^\uparrow = \RealNum^d \rtimes \LorentzGroup_+^\uparrow\)
  the elements of the proper orthochronous Poincaré group.
  In this paper we are working mostly with space-time
  translations of some given operator \(A\in \Alg(\Ww)\) by \(x \in \RealNum^d\), also denoted by
  \(\alpha_x(A)\).

Guided by physical intuition one asks that these objects satisfy wedge-local
variants of the Haag-Kastler postulates, which are concerned with the algebraic
and representation-theoretic properties of \(\Alg\).
Firstly, for any choice of wedge regions~\(\Ww, \Ww_1, \Ww_2 \) one has
\newcommand{\lt}[1]{\text{\bf #1}}%
{
\allowdisplaybreaks
\begin{align*}
  \lt{Isotony} \quad
    & \Alg(\Ww_1) \subset \Alg(\Ww_2) \text{ for } \Ww_1 \subset \Ww_2,
  \tag{HK1} \label{eq:HK1}\\
  \lt{Locality} \quad
  & \Alg(\Ww_1) \subset \Alg(\Ww_2)'
\text{ for } \Ww_1 \subset \Ww_2',
\tag{HK2} \label{eq:HK2}\\
\lt{Wedge-Duality}  \quad
  & \Alg(\Ww') = \Alg(\Ww)',
\tag{HK2$^\sharp$} \label{eq:HK2s}\\
  \lt{Translation-Covariance} \quad
  & \alpha_x(\Alg(\Ww)) = \Alg(\Ww +x), \quad x \in \RealNum^{d},\tag{HK3} \label{eq:HK3}\\
  \lt{Poincaré-Covariance} \quad
  & \alpha_\lambda(\Alg(\Ww)) = \Alg(\lambda\Ww), \quad \lambda
  \in \PoincareGroup_+^\uparrow. \tag{HK3$^\sharp$} \label{eq:HK3s}
\end{align*}}%
Here the Minkowski causal complement  \(\Ww' = (\Lambda \WwR + x)' = - \Lambda \WwR + x\) of
\(\Ww\) is also a wedge region and \(\Alg(\Ww)'\) denotes the commutant of
\(\Alg(\Ww)\) relative to \(\BoundedOps(\HilbertSpace)\).

On the representation-theoretic side it is further assumed that translations are
unitarily implemented on the vacuum Hilbert space~\(\HilbertSpace\) by a
strongly continuous \(\s\!+\!1\)-parameter group, \(\alpha_x(A) = U(x) A
  U(x)^*\). The representing unitaries are generated
  by the \term{energy-momentum operators}
  via~\(U(x) = U(t, \vec x) = \Ee^{\Ii t H - \Ii \vec x \cdot \Pp}\),
  whose joint spectral resolution in terms of projection-operator-valued measures
  will be denoted \( \Delta \longmapsto E_{(H,\Pp)}(\Delta)\) and abbreviated
  as \(E(\Delta)\) for Borel sets \(\Delta \subset \RealNum^{d}\).
Focusing   on the analysis of scattering, we impose the following standard assumptions concerned with the vacuum
representation and its one-particle spectrum,
\begin{align*}
  \lt{Uniqueness of \(\boldsymbol \Omega\)}
      \quad  & E(\{0\}) \HilbertSpace = \ComplexNum \Omega,
        \tag{HK4} \label{eq:HK4}\\
    \lt{Cyclicity of \(\boldsymbol \Omega\)}
      \quad & \closure{\Alg(\Ww)\Omega} = \HilbertSpace,
        \tag{HK5} \label{eq:HK5} \\
        \lt{Spectral Condition}
      \quad & \support E \subset \bar V^+,
        \tag{HK6} \label{eq:HK6}\\
    \lt{Haag-Ruelle Mass Gap Condition}
      \quad & H_m \subset \support E \subset \{0\} \cup H_m \cup \bar H_M,
        \tag{HK6$^\sharp$} \label{eq:HK6s}
\end{align*}
for some \(M > m > 0\), where
\(\bar V^+ := \{ (\omega, \vec p) : |\vec p| \leq \omega \}\)   denotes the positive
energy cone,
\(H_m := \{ (\omega_m(\vec p) , \vec p): \vec p
\in \RealNum^\s \}\),
\( \omega_m(\vec p):= \sqrt{\vec p^2 + m^2}\), is the (positive) hyperboloid of mass~\(m >
0\) and \(\bar H_M := \{ (\omega, \vec p): \vec p
\in \RealNum^\s, \omega \geq \omega_M(\vec p) \}\) denotes the convex hull of
\(H_M\).
\begin{Def}
  \label{def:wqft}
  A wedge-local quantum field theory is a tuple \((\Alg, \alpha, \HilbertSpace,
  \Omega)\) as above, which satisfies the basic assumptions
  \eqref{eq:HK1}--\eqref{eq:HK6}.
\end{Def}

This summarizes the axiomatic operator-algebraic formalism for wedge-local QFTs
providing the basis of our present investigations. Our choice of framework serves the purpose
of accommodating scattering-theoretic reasoning and capturing the requirements of the
wedge-local Haag-Ruelle theory~\cite{MD17a}. In the
literature also the closely related concept of a {\em causal Borchers triple}~\(({\mathcal
R}, U, \Omega) \), corresponding to \({\mathcal R} := \Alg(\WwR)\) in the
present framework, is often studied. Such a
framework is equivalent to the above setting including Poincaré
covariance~\eqref{eq:HK3s}.  For further historical background and 
other aspects of wedge-local Quantum Field Theory, we refer to
\cite{BS08,Le15,BLS10,MD17a}.

Regarding the scattering-theoretic analysis,
Poincaré covariance~\eqref{eq:HK3s} is not essential.
The important properties for us are the existence of isolated mass shells \eqref{eq:HK6s} and the
well-established wedge duality condition~\eqref{eq:HK2s},
which strengthen \eqref{eq:HK6} and \eqref{eq:HK2}, respectively.
As such, these will be standing assumptions of this work.
Whereas
\eqref{eq:HK6} merely demands positivity of the Hamiltonian in any Lorentz
frame, the Haag-Ruelle mass gap condition \eqref{eq:HK6s} physically amounts
to non-triviality of the one-particle subspace~\(\HilbertSpace_1 :=
E(H_m)\HilbertSpace\). The upper mass gap~\(M > m\) has the same
technical purpose as in traditional Haag-Ruelle scattering theory: it enables the
efficient separation of one-particle states from the
remaining energy-momentum spectrum.

  \subsection{Warped Convolutions and the Grosse-Lechner model}

Our scattering-theoretic analysis is concerned with a general class of massive
wedge-local QFT models, which are constructed via a certain deformation method
introduced by Buchholz, Lechner, and Summers \cite{BLS10}. The Grosse-Lechner
models were introduced and studied from a non-commutative geometry perspective
in \cite{GL07}. Here we adopt the operator-algebraic approach from \cite{BS08,BLS10}
and study, in particular, the GL-models as constructed by applying
BLS-deformations to the free scalar field.

 The starting point of the BLS-construction is a general wedge-local model \((\Alg^0, \alpha, \HilbertSpace, \Omega)\).
 With \(\kappa\geq 0\) and, in dimension \(d=3+1\) additionally \(\eta \in
 \RealNum\), as parameters of the deformation, a
 family \(\Ww \longmapsto Q_\Ww\) of {\em warping matrices}\/ is now defined
 according to \cite{BS08,GL07} by either
\[
 Q_{\WwR} : = \begin{pmatrix}
   0 & \kappa & 0 &  0 \\
   \kappa & 0 & 0 &  0 \\
   0 & 0 & 0 &  \eta \\
   0 & 0 & - \eta  &  0 \\
 \end{pmatrix}, \; (s = 3), \qquad
  Q_{\WwR} := \begin{pmatrix}
   0 & \kappa & 0 &\cdots &  0 \\
   \kappa & 0 & 0 &\cdots &  0 \\
   0 & 0 & 0 &\cdots &  0 \\
   \vdots & \vdots & \vdots & \ddots & \vdots \\
   0 & 0 & 0 & \cdots & 0
\end{pmatrix}, \; (\text{for general } s \geq 1),
 \label{eq:warpedq}\numberthis
\]
for the warping matrix of the reference wedge \(\WwR = \{ x \in \RealNum^{\sPone}:
  \abs{x^0} < x^1\}\) and its translates.  For general wedges given by \(\Ww = \Lambda
\WwR + x\) the warping matrix is defined by Poincaré-covariance\footnote{We
write \(Q\) as linear map, or \((1,1)\)-tensor~\(Q = (Q^\mu\,_\nu)_{\mu\nu}\), as
in \cite{BS08}. In \cite{GL07} and other works, especially in contexts of non-commutative
spacetime, writing \(Q_\WwR\) in \((2,0)\) or \((0,2)\)-tensor notation seems more common.
} as \(Q_{\Ww} := \Lambda
Q_{\WwR} \Lambda^{-1}\). These warping matrices are antisymmetric with respect
to the scalar product defined by the Lorentzian metric~\(g\)
with signature \((+,-,\ldots,-)\).

Let now \(Q \in \RealNum^{d\times d}\).
The {\em warped convolution} of an operator \(A\) is denoted by
\[
  A_Q := \int \DInt E_{(H, \Pp)} (p) \; \alpha_{Qp}(A),
  \numberthis
  \label{eq:warpright}
\]
and usually intuitively but formally written as a spectral integral with
operator-valued integrand.  Buchholz, Lechner and Summers \cite{BLS10} have
given a rigorous definition by interpreting \eqref{eq:warpright} as an
integral of oscillatory-type.
To this end, one starts with operators~\(A\) on which space-time translations act
smoothly and defines \(A_Q\) in terms of its action on vectors from the dense domain
\[
  \FiniteEnergyStates 
  :=\bigcup_{\substack{\Delta \subset \RealNum^d\\\text{compact}}}
  \!\!
  E_{(H, \Pp)}(\Delta) \HilbertSpace
  \numberthis
\]
of finite energy states.  We denote by \(\Cinfty\) the algebra of all {\em
  regular}\/ operators \(A \in \BoundedOps(\HilbertSpace)\), defined by the
  requirement that 
  \(\RealNum^{d} \ni x\longmapsto \alpha_x(A) = U(x) A U(x)^*\)
is arbitrarily often differentiable with respect to the operator norm topology.
By standard mollification arguments, the regular
subalgebra \(\Alg^{0r}(\Ww) := \Cinfty \cap \Alg^0(\Ww)\), for any wedge
\(\Ww\),
is weakly dense in the full wedge algebra~\(\Alg^0(\Ww)\) of the initial model.

\begin{Def}[warped convolution of a regular operator \cite{BLS10}] 
  \label{def:wc}
  The warped convolution~\(A_Q\) of a bounded regular operator~\(A \in
  \Cinfty\) with
  respect to a warping matrix \(Q\) is defined by its action on finite-energy
  vectors \(\Psi \in \FEnV\) as limit of strong integrals of oscillatory type,
  \[
    A_Q \Psi
    := \lim_{\epsilon \to 0} \frac{1}{(2\pi)^{\sPone}}
    \int \DInt[\sPone]x \, \DInt[\sPone]y \, \eta(\epsilon x, \epsilon y)  \,
    \Ee^{- \Ii x \cdot y} U(x) \alpha_{Qy}(A) \Psi
    \label{eq:defwc}\numberthis
  \]
  with~\(\eta \in \SchwartzSpace(\RealNum^{\sPone} \times \RealNum^{\sPone})\),
  such that \(\eta(0,0) = 1\), serving as a regularizing function.
\end{Def}

Here the technical assumptions of regularity of \(A\) and the consideration of
only finite energy states are needed for the existence of the limit
\(\epsilon \to 0\).
In addition to existence of this limit, it is shown in  \cite{BLS10} that \(A_Q\)
extends to a bounded operator, and is independent of the choice of~\(f\).
Let us summarize some helpful properties of warped convolutions from the
literature. For standard results about Bochner integrals we refer to
\cite[Ch.~6 §31]{Zaa67}.

\begin{Lem}{\rm \cite{BLS10}}
  \label{thm:wc}
  The warped convolution \(A_Q\) of any regular operator \(A\in \Cinfty\), given
  by \eqref{eq:defwc}, is well-defined. In particular it does not depend on the specific choice of the
  regularizing function~\(f\). %
  The deformed operator extends to a bounded operator
  \(A_Q \in \BoundedOps(\HilbertSpace)\), which satisfies
\begin{enumerate}[(i)]
  \item \( A_Q \Omega = A\Omega \),
    \label{it:thmwcvac}
  \item
    \((A_Q)^* = (A^*)_Q  \),
  \item \(A_0 = A\), \label{it:trivdeform}
\item
\([A_Q, A'_{-Q} ] = 0\),
whenever \(Q V^+ \subset \WwR\), \((-Q) V^+ \subset \WwL\), \(A \in \AlgR(\WwR)\), and \(A' \in \AlgR(\WwL)\). 
\item \label{it:deformcov} Warping commutes with space-time translations,
  \[
    \alpha_{x}(A_Q) = (\alpha_x(A))_{ Q}.
  \]
\item If the underlying model satisfies Poincaré-covariance \eqref{eq:HK3s} one has further
  \[
    \alpha_{\lambda}(A_Q) = (\alpha_\lambda(A))_{\Lambda Q \Lambda^{-1}}, \;
    \text{for }\lambda=(\Lambda,x) \in \PoincareGroup_+^\uparrow.
  \]
\end{enumerate}
\end{Lem}
\noindent

It is important to emphasize that the definition \eqref{eq:defwc} of \(A_Q\) is
only mathematically rigorous when working with regular operators \(A \in \AlgR\).
To obtain the full warped wedge-local QFT model in the
sense of \Cref{sec:prelwqft}, we additionally pass to the weak closure
of the set of warped regular operators.

\begin{Def}[warped wedge-local model \cite{BS08,BLS10}]
  Let \(Q_\Ww\) be given by \eqref{eq:warpedq} for \(\Ww = \WwR\)  and let \(Q_\Ww :=
  \Lambda Q_{\WwR} \Lambda^{-1}\) for general wedges \(\Ww = \Lambda \WwR + a\).
  The warped wedge-local algebras are defined as
  \[
    \Alg^Q(\Ww)
    :=  \{ A_{Q_\Ww}: \; A \in \Alg^{0r}(\Ww)\}'' \subset
    \BoundedOps(\HilbertSpace),
    \numberthis \label{eq:defalg}
  \]
  where the bicommutant is taken with respect to \(\BoundedOps(\HilbertSpace)\).
\end{Def}

\begin{Thm}[\cite{BLS10} Thm.\ 4.2] \label{thm:wcth}
 \((\Alg^Q, \alpha, \Omega, \HilbertSpace)\) defines a wedge-local QFT model
 satisfying the wedge-local Haag-Kastler postulates
 \eqref{eq:HK1}--\eqref{eq:HK6} from the initial model \((\Alg, \alpha,
 \Omega, \HilbertSpace)\).  Further, the deformed model satisfies the Haag-Ruelle
 spectral condition \eqref{eq:HK6s}, 
 wedge-duality \eqref{eq:HK2s}, or
 Poincaré covariance \eqref{eq:HK3s},
 respectively, if and only if the initial model satisfies
 the respective conditions.
\end{Thm}

We further note that \cite{BLS10} study wedge-local quantum field theories
  in the framework of a {\em (causal) Borchers
  triple}~\((\Alg(\WwR),\alpha,\Omega)\). There one specifies
  only an observable algebra for the single wedge region~\(\WwR\). The full family of wedge
  algebras is then obtained using Poincaré symmetry~\eqref{eq:HK3s}. This
  symmetry of the wQFT model is of course an additional requirement,
  which is natural from the perspective of physics and a useful constraint for
  constructive efforts. In particular, it is a non-trivial feature that BLS-deformations 
  can be used to construct Poincaré covariant models.
  However, the BLS-deformation method also applies to the present framework,
  where Poincaré covariance is optional. The proofs of \Cref{thm:wc} and
  \Cref{thm:wcth} can be extracted without significant modifications from the
  results of \cite{BLS10}.

\subsection{Scattering States in Wedge-local QFT} \label{sec:framewScatt}

We briefly review the \(n\)-particle Haag-Ruelle scattering  theory for
wedge-local QFTs \cite{MD17a}.  To construct
one-particle states via the wedge-local Haag-Ruelle method, one chooses
operators~\(A\in \Alg(\Ww)\) such that \(E(H_m)A\Omega \not = 0\).
The existence of such operators follows from \eqref{eq:HK5}.
Let \(\chi \in \SchwartzSpace(\RealNum^d)\) be a Haag-Ruelle auxiliary function
supported within a sufficiently small neighborhood of the mass shell~\(H_m\),
disjointly from the remaining energy-momentum spectrum. 
Then the one-particle problem is solved by
\[
  B := A(\chi) := \int \DInt[d] x \, \chi(x) \, \alpha_x(A) . \label{eq:aux}\numberthis
\]
Namely, in terms of the Fourier transform defined in the relativistic unitary
convention by
\[
  \hat \chi(p) = \hat\chi(p^0, \vec p) = \int \frac{\DInt[d]x}{(2\pi)^{d/2}} \,
  \Ee^{\Ii p^0 x^0 - \Ii \vec p \cdot \vec x} \, \chi(x^0, \vec x),
  \label{eq:defft}\numberthis
\]
we have \(B\Omega = (2\pi)^{d/2} \hat \chi(H, \Pp) A
\Omega\) by spectral calculus and translation-invariance of the vacuum. Hence
\(B\Omega\) is in the one-particle
space~\(\HilbertSpace_1:=E(H_m)\HilbertSpace\), as a consequence of the support
of the Haag-Ruelle auxiliary function~\(\hat \chi\) intersecting the energy-momentum
spectrum~\(\support E_{(H,\Pp)}\) of the theory only on subsets of the mass shell \(H_m\) (by
construction and using also the mass gap assumption \eqref{eq:HK6s}, see e.g.\
\cite[Ch.\ 5]{ArQFT99} or \cite[Sec.\ 2.1]{Dy17}).

Proceeding towards the \(n\)-particle problem we define similarly for a family
of \(A_j \in \Alg(\Ww)\), \(1 \leq j \leq n\), the operators \(B_j := A_j(\chi)\).
Further we consider (positive-energy) Klein-Gordon solutions
\begin{align*}
  f_j(t, \vec x) &= \int \frac{\DInt[s] k}{(2\pi)^\s} \,
  \Ee^{\Ii \vec k \cdot \vec x - \Ii \omega_m(\vec k) t} \;
  \tilde f_j(\vec k),  \quad (1 \leq j \leq n),
  \label{eq:defKG}\numberthis
\end{align*}
as scattering-theoretic comparison dynamics,
with relativistic dispersion \(\omega_m(\vec k) := \sqrt{\vec k^2 + m^2}\) for
mass \(m>0\), and recall that we abbreviate \(s := d-1\).  For usual technical
reasons, the {\em wave packets} are assumed to be {\em regular}, that is, \( \tilde
  f_j \in \ContinuousFuncs^\infty_c(\RealNum^{\s})\).
  By a stationary phase
analysis, these regular Klein-Gordon solutions  vanish rapidly in all space- and
time-like directions away from the classical propagation cones
\[
  \Upsilon_{f_j} := \{ (t, \vec x) \in \RealNum^{\sPone}:  \exists \vec k_j \in \support
    \tilde f_j: t \, \vec k_j = \omega_m(\vec k_j) \, \vec x
  \}. \numberthis \label{eq:propcone}
\]
We note that these cones describe the scattering geometry of the single-particle
wave packets.

The construction of \(n\)-particle states is now accomplished by means of
wedge-frame adapted Haag-Ruelle creation-operator approximants
\begin{align*}
   \BB_{j,\T}^\Lambda(f_j) : = \int \DInt[s] x \; f_j(\Lambda(\T, \vec x))
   \, \alpha_{(\Lambda(\T, \vec x))} (\Bb_j), \quad (\T \in \RealNum).
    \label{eq:defHRgen} \numberthis
\end{align*}
Here, the boost \(\Lambda \in \LorentzGroup_+^\uparrow\) specifies an 
auxiliary Lorentz frame used for the construction.
The simplest case is \(A_j \in \Alg(\Ww)\) for \(\Ww = \pm \WwR\) in two dimensions, and in
higher dimensions also for spatial rotations of \(\WwR\). Then one can simply take \(\Lambda =
\Id\), corresponding to equal-time hyperplanes. For general wedges~\(\Ww\) it is
technically preferable to choose this boost from
\[
  \LorentzGroup^*(\Ww) := \{ \Lambda \in \LorentzGroup_+^\uparrow : \Lambda \WwR = \Ww_c \}.
    \numberthis
\]
Regarding the causal geometry of wedges, we recall here that in dimension \(d \geq 2+1\)
a general wedge region~\(\Ww\) can be written as \(\Ww = \lambda \WwR : = \Lambda \WwR +
x\) in terms of the standard Rindler wedge \(\WwR\) and some Poincaré
transformation \(\lambda=(x,\Lambda)\) with Lorentz transformation part \(\Lambda \in \LorentzGroup_+^\uparrow\)
and translation \(x \in \RealNum^{d}\). Dimension~\(d = 1+1\) is special in
that, after restricting to the proper orthochronous Poincaré group, the set of
all wedges splits into the disjoint orbits of the right wedge~\(\WwR\) and
left wedge~\(\WwL:= \WwR' = - \WwR\).

The choice of \(\Lambda\) enters in the space-time
localization of the operators~\(\BB_{j,\T}^\Lambda(f_j)\), and keeping track of
the latter is important for the Haag-Ruelle method.
Namely, the localization depends on \(\T\) and
the wedge~\(\Ww\) of localization of~\(\Aa_j\), translated
to
\[
  \Ww + \T \VS_{f_j}^\Lambda \subseteq \RealNum^{d}, \quad
  \VS_{f_j}^\Lambda := \Upsilon_{f_j} \cap \Lambda T_1, \; \Lambda \in
  \LorentzGroup_+^\uparrow.
    \label{eq:modVeloSupp}\numberthis
\]
Here \(T_1 := \{ (1,\vec x): \vec x \in \RealNum^s \}\) is the standard
space-like hyperplane at \(\T = 1\) and \(\VS_{f_j}^\Lambda\) is the {\em velocity
support} of \(f_j\) with respect to the Lorentz frame specified by
\(\Lambda\). For further details we refer to \cite{MD17a}}.

Geometrically, scattering situations are concerned with phenomena at
very large distances,
and it is convenient to introduce the centering of \(\Ww\)
denoted by
\(\Ww_\C := \Lambda \WwR\).
As centered wedges~\(\Ww_\C\) are convex cones with \(\Ww_\C + \Ww_\C \subset \Ww_\C\),
the \term{precursor} relation
\[
  \Reg_1 \prec_\Ww \Reg_2 :\Longleftrightarrow \Reg_2 - \Reg_1 \subset \Ww_\C,
  \label{eq:defPrec}\numberthis
\]
defined for non-empty regions \(\Reg_1,\Reg_2 \subset \RealNum^d\), is transitive,
and Poincaré covariant as a partial ordering in the sense that
\(
  \Reg_1 \prec_\Ww \Reg_2 \Longleftrightarrow \lambda \Reg_1 \prec_{\Lambda \Ww}
  \lambda \Reg_2, \; \lambda = (x, \Lambda) \in \PoincareGroup_+^\uparrow
\).

To construct \(n\)-particle scattering states we consider ordered configurations
of wave packet velocities
\[
  \VS_{f_n}^\Lambda \prec_\Ww \VS_{f_{n-1}}^\Lambda \prec_\Ww \ldots
  \prec_\Ww \VS_{f_1}^\Lambda,
  \label{eq:ordPrescL}\numberthis
\]
for the outgoing limit \(\T \to + \infty\), and the reversed ordering
\[
  \VS_{f_n}^\Lambda \succ_\Ww \VS_{f_{n-1}}^\Lambda \succ_\Ww \ldots \succ_\Ww \VS_{f_1}^\Lambda,
  \label{eq:ordPrescIn}\numberthis
\]
for the incoming scattering limit \(\T \to -\infty\), respectively.
   In these cases the main result of \cite{MD17a} establishes convergence
 of the scattering-state approximants
\[
  \Psi_n^\Lambda(\tau) :=  \Bb_{1\T}^\Lambda(f_1) \ldots \Bb_{n\T}^\Lambda(f_n) \Omega,
  \label{eq:scattApprox}\numberthis
\]
and Fock structure of the limits, where all underlying wedge-local operators
\(A_j \in \Alg(\Ww)\) are localizable in a common wedge \(\Ww\).  This approach
is distinct from conventional Haag-Ruelle theory, where convergence and Fock
structure proofs use that the commutators~\([\Bb_{j\T}^\Lambda(f_j),
\Bb_{k\T}^\Lambda(f_k)]\) vanish rapidly in norm  for \(j \not = k\) when \(\T
\to \pm \infty\).  Such stronger estimates are obtained from locality of the QFT
and disjoint velocity supports of the Klein-Gordon solutions.  Applying this
reasoning in a wedge-local context leads to a more restrictive setup. Firstly we
should take \(A \in \Alg(\Ww)\) and \(A^\perp \in \Alg(\Ww^\perp)\) for an {\em
opposite wedge}~\(\Ww^\perp := \Ww' + x\), \(x\in \RealNum^d\). Additionally,
the Klein-Gordon solutions~\(f, f^\perp\) must be chosen such that the
velocity supports are not merely disjoint, but satisfy the stronger
geometrical ordering property
\[
  \VS_{f^\perp}^\Lambda \prec_\Ww \VS_{f}^\Lambda.  \numberthis
\]
Then one obtains by wedge-locality and the decay properties described in
\eqref{eq:propcone} and \eqref{eq:modVeloSupp} a rapid decay
for large \(|\T|\),
\begin{align*}
  \norm{[\Bb^{\Lambda}_\T(f), \Bb^{\perp \Lambda}_\T(f^\perp)]}
  & \leq \frac{C_N}{\T^N}, \qquad \; (\T > 0),
     \numberthis \label{eq:commout}
  \\
  \norm{[\Bb^{\Lambda}_\T(f^\perp), \Bb^{\perp \Lambda}_\T(f)]}
  & \leq \frac{C_N}{(-\T)^N}, \quad (\T < 0). \numberthis
\end{align*}
Here it should be especially noted that the geometrical configuration depends
on whether the outgoing (\(\T > 0\)) or incoming (\(\T < 0\)) regime is
considered.
Pairs of opposite wedge configurations were used in the previous constructions
of two-particle scattering states in the wedge-local context \cite{GL07,BS08}, and also
using wedge-local operators in local QFT context \cite{BBS01}.

In the multi-particle generalization~\eqref{eq:scattApprox} opposite wedges
appear only indirectly. Namely, we work with one-particle states that can be
generated from the vacuum within two opposite wedges \(\Ww\), \(\Ww^\perp =
\Ww' + x\) (for some \(x \in \RealNum^d\)),
\[
   \Psi = A \Omega = A^\perp \Omega, \quad A\in \Alg(\Ww), \; A^\perp \in
   \Alg(\Ww^\perp).
   \label{eq:swap}\numberthis
\]
In this case we call \(\Psi\) {\em swappable} (with
respect to the wedge \(\Ww\)). Swappable vectors with \(\Ww^\perp := \Ww'\)
are dense in the full Hilbert space~\(\HilbertSpace\) as a consequence of
wedge duality \eqref{eq:HK2s}, see \cite[App.~B]{MD17a}.
Swappable one-particle states are obtained by projecting swappable vectors onto
the one-particle space \(\HilbertSpace_1\). In this way we obtain
from \eqref{eq:aux} and \eqref{eq:defHRgen} an oppositely localized pair of Haag-Ruelle
creation-operator approximants for each one-particle state (\(1\leq k \leq
n\))
\[
  \Psi_1^k := \Bb^{\Lambda}_{k\T}(f_k)\Omega = \Bb^{\perp
  \Lambda}_{k\T}(f_k)\Omega, \numberthis
\]
where both expressions involving Haag-Ruelle operators are \(\T\)-independent by construction, so that
\(\T \to \pm \infty\) limits can be dropped. Although scattering states can be
constructed without explicit use of the swapped operators~\( \Bb^{\perp
\Lambda}_{k\T}(f_k)\), they are the main tool for establishing the wedge-local
Haag-Ruelle theorem:
\begin{Thm}[\cite{MD17a}]\label{thm:hrwGen} i
  Let \((\Alg, \alpha, \HilbertSpace, \Omega)\) be a  wedge-local quantum
  field theory satisfying  wedge duality \eqref{eq:HK2s} and the mass gap
  condition~\eqref{eq:HK6s}.  Let \(\Lambda \in \LorentzGroup^\uparrow_+\) and
  \(\Psi_1^j = E(H_m) A_j \Omega\) with \(A_j \in \Alg(\Ww)\) be swappable one
  particle states, and define \(B_k:=A_k(\chi)\) with an auxiliary
  function~\(\chi\) as in \eqref{eq:aux}. 
  \begin{enumerate}[(i)]
    \item
      \label{it:hrwgconv} For regular positive-energy Klein-Gordon solutions
        \(f_j\) satisfying
        \[ 
          \VS_{f_n}^\Lambda \prec_\Ww \VS_{f_{n-1}}^\Lambda
          \prec_\Ww \ldots \prec_\Ww \VS_{f_1}^\Lambda,
          \label{eq:ordPrescLhrw}\numberthis
        \] 
        the scattering state approximants \(\Psi^\Lambda_n(\T) :=
          \BB_{1\T}^\Lambda(f_1) \BB_{2\T}^\Lambda(f_2) \ldots \BB_{n\T}^\Lambda(f_n)\Omega \)   
          converge rapidly in norm for \(\T\to \infty\).
        More precisely,  for any \( N \in \NaturalNum\) there exists a \(C_N > 0\) such that
        \[
           \norm{\Psi_n^\Lambda(\tau) - \Psi^{+,\Lambda}_n} \leq \frac{C_N}{\tau^N},
           \quad (\T > 0).
           \numberthis
        \]

  \item \label{it:hrwgfock} For \(\Lambda \in \LorentzGroup^*(\Ww)\) scalar
    products of \(\Psi_n^{+,\Lambda} := \lim_{\T\to\infty} \BB_{1\T}^\Lambda(f_1)
      \ldots \BB_{n\T}^\Lambda(f_n) \Omega \), \(\Psi'^{+,\Lambda}_{n'}
      := \lim_{\T\to\infty} \BB_{1\T}'^\Lambda(f_1') \ldots \BB_{n'\T}'^\Lambda(f_{n'}') \Omega \), 
      constructed both with respect to the same wedge~\(\Ww\), satisfy
      \[ 
        \left\langle \Psi^{+,\Lambda}_n, \Psi'^{+,\Lambda}_{n'}\right \rangle 
        = \delta_{nn'} \prod_{j=1}^n 
            \left\langle \BB_{j\T}^\Lambda(f_j)\Omega, \BB_{j\T}'^\Lambda(f_j')\Omega \right\rangle. 
        \label{eq:fockGen}\numberthis
      \]
    Here the right hand side is again \(\tau\)-independent by construction. 
  \end{enumerate}
  Analogous statements hold for incoming scattering states, assuming opposite
  ordering.
\end{Thm}
For the later discussion of the scattering matrix, let us note that all
asymptotic data in wedge-local models, including wave operators, must be defined
depending on a localization wedge~\(\Ww\) from which the scattering states have
been prepared. This is an unusual feature of wedge-local models.
But we note that, of course, the \(\Ww\)-dependence can be trivial. This happens
in \(1\!+\!1\) dimensions, where our results apply as well. In this case there are
only two centered wedges~\(\WwR\) and \(\WwL\), whose associated observables and
scattering states can be related by swapping symmetry. In general this certainly does not
imply that the models will be local. Presently the status of existence or
non-existence of local observables even in \(1\!+\!1\)-dimensional GL-models appears
to be still open.
In the present formalism for higher dimensions, the existence of a local QFT
model underlying the wedge-local model under consideration implies a certain
trivial \(\Ww\)-dependence of scattering states, which is discussed in \Cref{sec:AppGL}.

The possible more general wedge-dependence of scattering states is an interesting feature of
wedge local quantum field theories in higher dimensions. In particular, we will use the
formalism introduced in~\cite{MD17a} to describe this wedge dependence in a more
transparent manner for two-particle and \(n\)-particle scattering reactions. One
of the main aims of the present paper is to illustrate this wedge dependence in
BLS-deformed wQFT models and in particular for the special case of the
Grosse-Lechner models (see Sections~\ref{sec:ScattDat} and \ref{sec:AppGL}).

\section{\(N\)-Particle Scattering in BLS-Deformed wQFTs} \label{sec:GL}

Using the results and notation described in the previous sections, we can now state our
results in precise form.  Let \((\Alg^0, \alpha, \HilbertSpace, \Omega)\) be a
given  wedge-local quantum
field theory satisfying  wedge duality \eqref{eq:HK2s} and the mass gap
condition \eqref{eq:HK6s}, in addition to \eqref{eq:HK1}--\eqref{eq:HK6}.
Let \((\Alg^Q, \alpha, \HilbertSpace, \Omega)\) be the model constructed by
BLS-deformation with some fixed warping parameter~\(Q\).

Our main aim will be to prove asymptotic completeness of the deformed model.
On the technical side, this will be achieved by establishing a direct relation
between a scattering state
 \[
  \lim_{\T \to \pm \infty}  \BB_{1,\T}^{Q,\Lambda}(f_1) \BB_{2,\T}^{Q,\Lambda}(f_2)
  \ldots \BB_{n,\T}^{Q,\Lambda}(f_n)\Omega \label{eq:techoutdef} \numberthis
\]
of the deformed model, and the scattering states of the initial model,
\[
  \lim_{\T \to \pm \infty}  \BB_{1,\T}^\Lambda(f_1) \BB_{2,\T}^\Lambda(f_2)
  \ldots \BB_{n,\T}^\Lambda(f_n)\Omega. \label{eq:techoutorig} \numberthis
\]
Here we write \(B_k^Q = A_k^Q(\chi)\) with \(\chi \in
\SchwartzSpace(\RealNum^d)\), as before, and \(A_k^Q \in \Alg^Q(\Ww)\) for
\(1\leq k \leq n\) with some fixed centered wedge \(\Ww = \Lambda \WwR\). In the
following we will for simplicity suppress dependencies on the fixed reference
frame boost~\(\Lambda\) and on the wedge in our notation, writing for example
\(\prec\), \(\VS_{f_k}\) and \(B_{k,\T}^Q(f_k)\), instead of \(\prec_\Ww\),
\(\VS^\Lambda_{f_k}\) and
\(B_{k,\T}^{Q,\Lambda} (f_k)\), etc., when these dependencies can be clearly seen
from the context.

\subsection{Wave Operators in BLS-deformed wQFTs}

Our introductory heuristic considerations (see equations~\eqref{eq:psiqdef} and
\eqref{eq:psi0trans}) suggested that the deformation preserves outgoing and
incoming particle numbers. This makes our approach feasible. As we can in
principle deform general models, such as \(P(\phi)_2\), the scope of our result
is not restricted to particularly simple or integrable deformed models. However, we
can infer similarly that the deformed scattering state \eqref{eq:techoutdef}
cannot be directly related to a scattering state in the
undeformed model with the same simple product form~\eqref{eq:techoutorig}. This is addressed here by using the formulation of
the scattering data in terms of wave operators~\(\WO^\pm_\Ww\), as defined in \cite[Sec.~5]{MD17a}.
To recall the formal construction of the wave
operators we start by introducing the full (unsymmetrized)
Fock space by
\begin{align}
  \Gamma^u(\HilbertSpace_1) : = \bigoplus_{n=0}^\infty
  {\HilbertSpace_1}^{\otimes n},
\end{align}
over the one particle space
\(\HilbertSpace_1 := E_{(H,\Pp)}(H_m)\HilbertSpace\).
Here it is convenient to define the velocity ordering
\(\succ_\Ww\) also as partial order on \(\HilbertSpace_1\), by writing
\[
  \Psi_1 \prec_\Ww \Psi_1' : \Longleftrightarrow
  \VS_{\Psi_1}^\Lambda \prec_\Ww \VS_{\Psi_1'}^\Lambda,
  \label{def:veloOrdOPS} \numberthis
\]
where the velocity support~\(\VS_{\Psi_1}^\Lambda\) of a one-particle
vector~\(\Psi_1\) is defined analogously to \eqref{eq:propcone} and \eqref{eq:modVeloSupp},
replacing the support of \(\tilde f_j\) by the spectral support of the momentum
operator in the state~\(\Psi_1\).
Natural domains for the wave operators in the wedge-local setting are the
velocity-ordered Fock spaces~\(\Gamma^{\succ}(\HilbertSpace_1)\) and
\(\Gamma^{\prec}(\HilbertSpace_1)\),  which are defined as the closures of
the spans
\begin{align}
  \Gamma^{\succ}_0(\HilbertSpace_1)
  &:= \Span{ \{
    \Psi^1\otimes \ldots \otimes \Psi^n : n \in \NaturalNum_0, \;
    \Psi^1,\ldots, \Psi^n  \in \HilbertSpace_1 ,
    \; \Psi^1 \succ \Psi^2 \succ \ldots \succ
\Psi^n \}}, \\
  \Gamma^{\prec}_0(\HilbertSpace_1)
  &:= \Span{ \{
    \Psi^1\otimes \ldots \otimes \Psi^n : n \in \NaturalNum_0, \;
    \Psi^1,\ldots, \Psi^n  \in \HilbertSpace_1,
    \; \Psi^1 \prec \Psi^2 \prec \ldots \prec \Psi^n \}},
\end{align}
of finite linear combinations, 
for outgoing- and incoming scattering states, respectively. 
For technical purposes we analogously define such algebraic spans
\(\Gamma^\succ_0(\HilbertSpace_1')\)
and \(\Gamma^\prec_0(\HilbertSpace_1')\) for dense subsets \(\HilbertSpace'_1
\subset \HilbertSpace_1\).

Our main technical results concern the wave operators~\(\WO^{\pm}_{Q,\Ww}\) of
the deformed model, and we will now discuss their construction as given in
\cite{MD17a}. Here we will also include some additional details which are of
particular importance for us here. First we recall that the wave operators are
constructed by means of \Cref{thm:hrwGen}. In particular this means that, on the
technical side, we are in fact working with a smaller subset of one-particle
states, namely all which can be written in the form
\[
  \Psi_1 =  \BB_{\T}^\Lambda(f) \Omega  = \BB_{\T}^{\perp\Lambda}(f) \Omega.
   \label{eq:swaponetech} \numberthis
\]
We call such \(\Psi_1\) swappable (with respect to \(\Ww\)) one-particle states of bounded
energy, and we let \(\HilbertSpace_{1c}^\Ww\) be the (non-closed) linear space
spanned by them.
Alternatively, using spectral calculus, one sees that the above generating
set\footnote{A general vector \(\Psi_1 \in \HilbertSpace_{1c}^\Ww\) is a sum of vectors
  of the form \eqref{eq:swaponetech} or \eqref{eq:momresdir}.} of \(\Psi_1 \in
\HilbertSpace_{1c}^\Ww\) can also be characterized by the
existence of \(A \in \Alg(\Ww)\), \(A^\perp \in \Alg(\Ww^\perp)\), and \(\tilde f
  \in \ContinuousFuncs^\infty_c(\RealNum^s)\), such that
  \[\Psi_1 = \tilde f(\Pp)
  E_{(H,\Pp)}(H_m)  A \Omega = \tilde f(\Pp) E_{(H,\Pp)}(H_m) A^\perp \Omega.
  \label{eq:momresdir} \numberthis
\]

\begin{Prop} \label{prop:denseSOS}
In a wQFT satisfying wedge duality~\eqref{eq:HK2s}, the density
\(\closure{\HilbertSpace_{1c}^\Ww} = \HilbertSpace_1\) holds for any wedge \(\Ww\).
\proof  Let \(\Psi_1 \in \HilbertSpace_1\) and \(\epsilon  > 0\).
By
\cite{MD17a}~App.~B,
there exist \(A \in \Alg(\Ww)\) and \(A^\perp \in \Alg(\Ww')\), such that
\(\Psi' := A\Omega = A^\perp \Omega\) and \(\norm{\Psi' - \Psi_1} < \epsilon/2\).
Let now \(f \in \ContinuousFuncs^\infty_c(\RealNum^s)\) with \(\tilde f(0) = 1\). Then
for any \(\delta > 0\) we have \(\Psi_\delta := \tilde f(\delta \Pp) E(H_m)  \Psi'
  \in \HilbertSpace_{1c}^\Ww\) by \eqref{eq:momresdir}. By spectral calculus \( \Psi_\delta
  \longrightarrow E(H_m)\Psi'\) when we take \(\delta \to 0\).
  In particular 
  there exists \(\delta'> 0\) s.t.\  \( \|\Psi_{\delta'} - E(H_m)\Psi'\| <
  \epsilon/2\). Together we get
  \(
    \norm{\Psi_{\delta'} - \Psi_1}
    \leq
    \norm{\Psi_{\delta'} - E(H_m) \Psi'} +
    \norm{ E(H_m)(\Psi' - \Psi_1)}
    \leq
    \epsilon/2+ \norm{E(H_m)} \norm{ \Psi' - \Psi_1}
    \leq \epsilon
  \). \qed
\end{Prop}

\begin{Rem}
Concerning this technical detail, there is one subtlety with respect to the
standard construction of bosonic and fermionic Fock spaces, which should be pointed
out. Namely, we expect the resulting Fock space  to be independent of possibly
different choices of \(\HilbertSpace'_1\), as the choice of this dense subset of
the one-particle space is merely technically motivated.  Yet, for ordered Fock
spaces
the choices of \(\HilbertSpace'_1\) are not completely arbitrary due to 
a possible interplay with ordering conditions.
In the present context one can express this requirement on
\(\HilbertSpace_1'\) as follows: we
say that \(\HilbertSpace'_1 \subseteq \HilbertSpace_1\) is {\em momentum
resolving}
iff for any \(\Psi \in \HilbertSpace_1\) there exists a sequence
\((\Psi_n)_{n\in \NaturalNum} \subset \HilbertSpace'_1\) such that \(\Psi_n
\rightarrow \Psi\) and \(\support E_{(H,\Pp)}{\Psi_n} \subseteq \support
E_{(H,\Pp)}{\Psi}\). We note that from \Cref{prop:denseSOS} and \eqref{eq:momresdir} it
follows by choosing \(\tilde f\) suitably that the linear
spaces~\(\HilbertSpace^\Ww_{1c}\) are momentum resolving.
\end{Rem}

By construction, the energy-momentum operators of the initial and deformed model
coincide. In particular, the one-particle spaces of the two models are
identical.  Hence we can directly compare the scattering data of the two models
on the present abstract level in terms of the respective wave operators, as they
are defined on the same ordered Fock spaces for both the deformed and initial
model.
For a direct comparison of the two wave operators, we will make use of a strenthened
technical result on the density of one-particle states.

\begin{Cor}\label{prop:densereg}
  Any state \(\Psi_1 \in \HilbertSpace_{1c}^\Ww\) can be generated by 
  a pair \(\tilde A \in \Alg^r(\Ww)\) and \(\tilde A^\perp \in
  \Alg^{r}(\Ww^\perp)\) of regular operators.
  \proof Starting from a pair \(A \in \Alg(\Ww)\), \(A^\perp \in \Alg(\Ww')\)
  and regular wave packet \(\tilde f\)
  from \Cref{prop:denseSOS}  we write
  \[\Psi_1 = 
    \hat \chi(\omega_m(\Pp), \Pp)^{-1} \tilde f(\Pp)
  E_{(H,\Pp)}(H_m) \hat \chi(H, \Pp)   A \Omega ,
  \numberthis
  \]
  where \(\chi \in \ContinuousFuncs^\infty_c(\RealNum^d)\)
  is a compactly supported function with \(\hat \chi(\omega_m(\vec p),
  \vec p)\) non-vanishing for all \(\vec p \in \support \tilde f\).
  Then \(\tilde f'(\vec p) := 
  \hat \chi(\omega_m(\vec p), \vec p)^{-1} \tilde f(\vec p)\) defines a new regular
  Klein-Gordon wave packet and \(\hat \chi(H, \Pp)   A \Omega 
    = (2\pi)^{-d/2}  A(\chi) \Omega = : A' \Omega\)  (see eq.\ \eqref{eq:defft}) shows that the vector part is
  obtained from the vacuum as image of a regular operator \(A' \in \Alg^r(\tilde
    \Ww)\) for some wedge \(\tilde \Ww \supset \Ww + \support
  \chi\). An analogous calculation yields \(A'^\perp \in \Alg^r(\Ww^\perp)\) with
\(\Ww^\perp \supset \Ww' + \support \chi\).
We can then arrange \(\tilde A \in \Alg^r(\Ww)\) by translating the above constructed
operator back to \(\Ww\), and also replacing \(A^\perp\) and \(\Ww^\perp\) with
the respective translates by the same vector. These
translations can be compensated by absorbing a corresponding phase 
into the wave packet. \qed
\end{Cor}

\begin{Def}[wave operators of initial and deformed model] \label{def:wo}
  Let \((\Alg, \alpha, \HilbertSpace, \Omega)\) be a wedge-local quantum field
  theory.
  The incoming and outgoing wave operators~\(\WO^\pm_{0,\Ww}\) associated to centered wedge regions \(\Ww =
  \Lambda \WwR\), \(\Lambda \in \LorentzGroup\), in the initial model are the maps
  defined by
  \[
  \begin{aligned}
    \WO_{0,\Ww}^+ &:
    \left\{
      \begin{aligned}
      \FS^{\succ_\Ww}(\HilbertSpace_{1})  &\longrightarrow
           \HilbertSpace,\\
        \Psi_1^1 \otimes \ldots \otimes \Psi_1^n  &\longmapsto
     \lim_{\T \to\infty} \BB_{1\T}^\Lambda(f_1) \ldots \BB_{n\T}^\Lambda(f_n) \Omega,
    \end{aligned}
    \right.
    \\
    \WO_{0,\Ww}^- &:
    \left\{
      \begin{aligned}
        \FS^{\prec_\Ww}(\HilbertSpace_{1})  &\longrightarrow
        \HilbertSpace,\\
        \Psi_1^1 \otimes \ldots \otimes \Psi_1^n  &\longmapsto
     \lim_{\T \to-\infty} \BB_{1\T}^\Lambda(f_1) \ldots
           \BB_{n\T}^\Lambda(f_n) \Omega.
    \end{aligned}
    \right.
  \end{aligned}
  \label{eq:defWO}\numberthis
\]
via linear and continuous extension from product states in \(\Gamma_0^{\succ_\Ww/\prec_\Ww}(\HilbertSpace_{1c}^\Ww)\).
Similarly the wave operators of the deformed model~\((\Alg^Q, \alpha, \HilbertSpace, \Omega)\) are
\[
  \begin{aligned}
    \WO_{Q,\Ww}^+ &:
    \left\{
      \begin{aligned}
      \FS^{\succ_\Ww}(\HilbertSpace_{1})  &\longrightarrow
           \HilbertSpace,\\
        \Psi_1^1 \otimes \ldots \otimes \Psi_1^n  &\longmapsto
     \lim_{\T \to \infty} \BB_{1Q\T}^\Lambda(f_1) \ldots \BB_{nQ\T}^\Lambda(f_n) \Omega,
    \end{aligned}
    \right.\\
    \WO_{Q,\Ww}^- &:
    \left\{
      \begin{aligned}
      \FS^{\prec_\Ww}(\HilbertSpace_{1})  &\longrightarrow
           \HilbertSpace,\\
        \Psi_1^1 \otimes \ldots \otimes \Psi_1^n  &\longmapsto
     \lim_{\T \to -\infty} \BB_{1Q\T}^\Lambda(f_1) \ldots \BB_{nQ\T}^\Lambda(f_n) \Omega.
    \end{aligned}
    \right.
  \end{aligned}
  \label{eq:defWOQ}\numberthis
\]
\end{Def}

For simplicity we will refer to \(\WO^\pm_{Q,\Ww}\) as the {\em deformed wave
operators}. We note that, strictly speaking, this terminology is justified only in
retrospective after the results of the present paper are established. 
Namely, our main result shows the wave operators~\(\WO^{\pm}_{Q,\Ww}\), as defined
using the general construction in the deformed model, can be regarded as a
``deformation'' of the wave operators~\(\WO^{\pm}_{0,\Ww}\) from the underlying
``undeformed'' model.

\begin{Thm} \label{thm:defw}
  The wave operators of the deformed model~\((\Alg^Q, \alpha, \HilbertSpace, \Omega)\)
can be expressed in terms of the wave operators of the initial model~\((\Alg,
\alpha, \HilbertSpace, \Omega)\) via
\begin{align*}
  \WO_{Q,\Ww}^{+} &= \WO_{0,\Ww}^{+} S_{Q_\Ww}^{\succ_\Ww}, &
  \WO_{Q,\Ww}^{-} &= \WO_{0,\Ww}^{-} S_{Q_\Ww}^{\prec_\Ww},
  \label{eq:wootimesq}\numberthis
\end{align*}
where \(S_{Q_\Ww}^{\succ_\Ww/\prec_\Ww}\) are restrictions  to
\(\Gamma^{\succ_\Ww/\prec_\Ww}\)
of \(S_{Q_\Ww} : \FS^u(\HilbertSpace_1) \to \FS^u(\HilbertSpace_1) \)
defined
by
\[
  S_{Q_\Ww} \Psi^1_1 \otimes \ldots \otimes \Psi^n_1
  := \prod_{1\leq i < j \leq n} \Ee^{\Ii P_i\cdot Q_\Ww P_j}
\Psi^1_1 \otimes \ldots \otimes \Psi^n_1
\numberthis
\]
for \(\Psi_1^1, \ldots, \Psi_1^n \in \HilbertSpace_1\).
Here, for \(i \in \NaturalNum\), \(P_i\) denotes the self-adjoint operator defined on
\(\Gamma^u(\HilbertSpace_1)\) in terms of the energy-momentum operator \(P=
(H, \Pp)\) of the model by
\[
P_i (\Psi^1 \otimes \ldots \otimes  \Psi^n) =
   \Psi^1 \otimes \ldots \otimes P \Psi^i \otimes \ldots \otimes  \Psi^n
   \numberthis
\]
for \( n \geq i\),
\(\Psi_1, \ldots, \Psi_n \in \HilbertSpace_1\)
with \(\Psi^i \in \Domain{}(P)\) 
and \(P_i (\Psi^1 \otimes \ldots \otimes  \Psi^n) :=  0\) if \(n < i\).
\end{Thm}
\noindent
We note that \(P_i\) can be regarded as \(i\)-th unordered second quantization
of the energy-momentum operator~\(P\). Essential self-adjointness of \(P_i\) on
the domain of vectors of finite particle number follows from standard arguments
(see e.g.\ \cite[Sec.~VIII.10]{RS1}).

\subsection{Proof of the Wave Operator Identity}
We start with some preparations. 
First we will recall how to re-express the warped convolutions  
as convergent \(\HilbertSpace\)-valued integrals. This is done by using standard oscillatory
integral methods, similarly to the discussions from \cite{BLS10,DT11}. 

\begin{Lem}[oscillatory integral method] \label{lem:oscint}
  Let \(\RealNum^{2d} \ni (x,y) \longmapsto \Psi(x,y) \in \HilbertSpace\) be a
  map with uniformly bounded derivatives for all multi-indices~\(\beta\in
  \NaturalNum_0^{2d}\) up to order
  \(\abs{\beta} \leq 4d\).\footnote{Our strategy
    mostly follows \cite{BLS10}.  With our choice of \(D_\text{reg}\), integrability estimates
    are split into pairs \((x_j,y_j)\), \(1\leq j \leq d\), and thereby they
    become slightly more explicit. On the other hand, our smoothness
    requirements are not optimal. For a regularizing differential operator
    requiring only \(\abs{\beta} \leq
    d+1\) derivatives of \(\Psi(x,y)\), we refer to \cite{BLS10}.
  }
  Then we have
  \begin{align*}
    \lim_{\epsilon \to 0}
    \frac{1}{(2\pi)^{d}} \int \DInt[\sPone]x \; \DInt[\sPone]y \; \eta(\epsilon x,
    \epsilon y) \Ee^{- \Ii x \cdot y} \Psi(x,y)
    =
    \frac{1}{(2\pi)^{d}} \int \DInt[\sPone]x \; \DInt[\sPone]y \; \Ee^{- \Ii x \cdot y}
    D_{\text{reg}}(\partial_x, \partial_y) \Psi(x,y).  \label{eq:oscint}\numberthis
  \end{align*}
   Here \(D_{\text{\rm reg}}(\partial_x, \partial_y):= \prod_{j=0}^{s}
  D_j(\partial_{x_j},\partial_{y_j})^2 \) is a product of auxiliary,
  mutually commuting partial differential operators defined by
  \begin{align*}
    D_j(\partial_{x_j},\partial_{y_j}) \Phi(x,y) &:=
    (1 - \partial_{x_j}^2 - \partial_{y_j}^2) \frac{1}{1+y_j^2 + x_j^2} \;
    \Phi(x,y).
    \label{eq:mollj}\numberthis
  \end{align*}
 In particular, the limit \eqref{eq:oscint} exists and it is independent of the choice of the regularizing function \(\eta
    \in \SchwartzSpace(\RealNum^{\sPone} \times \RealNum^{\sPone})\) with
    \(\eta(0,0) = 1\).

  \proof The operators \(D_j\) are constructed such that their formal adjoints
  satisfy
  \[
    D_j^* \Ee^{\pm \Ii x_j y_j} =
    \frac{1}{1+y_j^2+x_j^2} (1 - \partial_{x_j}^2 - \partial_{y_j}^2) \Ee^{\pm \Ii x_j y_j} =
    \Ee^{\pm \Ii x_j y_j}. \numberthis
  \]
  Inserting this identity into the oscillatory integral with finite \(\epsilon >
  0\) twice for every \(1\leq j \leq d\), using integration by parts, and
  writing \(\eta_\epsilon(x,y) := \eta(\epsilon x, \epsilon y)\), we get
  \begin{align*}
    \fl
    \frac{1}{(2\pi)^{\sPone}}
        \int \DInt[\sPone]x \; \DInt[\sPone]y \; \Ee^{- \Ii x \cdot y}
        \eta_\epsilon(x, y) \Psi(x,y)
        \\&=
    \frac{1}{(2\pi)^{\sPone}}
        \int \DInt[\sPone]x \; \DInt[\sPone]y \; \Ee^{- \Ii x \cdot y}
        D_{\text{reg}}(\partial_x, \partial_y) \eta_\epsilon(x,y) \Psi(x,y).
        \numberthis
        \label{eq:intrewritten}
  \end{align*}
  In this form the \(\epsilon \to 0\) limit can now be carried out:
  it follows by explicit calculation and
  induction that
  \[
    \norm{D_\text{reg}(\partial_{x}, \partial_{y}) \Phi(x,y)} \leq
    \prod_{j=0}^{d-1}
    \frac{C}{(1+x_j^2+y_j^2)^2} \norm{\Phi}_{\mathscr C^{4d}(\RealNum^{2d},\HilbertSpace)},
    \label{eq:majest}\numberthis
  \]
  where the norm is defined by 
  \[
  \norm{\Phi}_{\mathscr C^{k}(\RealNum^{2d},\HilbertSpace)}
  := \sum_{\substack{\alpha,\beta \in \NaturalNum_0^d\\ \abs{\alpha}+
  \abs{\beta}\leq k}} 
    \sup_{x,y\in \RealNum^d} \norm{\partial_x^\alpha \partial_y^\beta \Phi(x,y)}_{\HilbertSpace}.
  \numberthis
\]
From \eqref{eq:majest} we now obtain an integrable majorant for
the rewritten integral~\eqref{eq:intrewritten}, by estimating 
\(
  \norm{\eta_\epsilon \Psi}_{\mathscr C^{4d}(\RealNum^{2d},\HilbertSpace)}
  \leq C
  \norm{\eta_\epsilon}_{\mathscr C^{4d}(\RealNum^{2d})}
  \norm{\Psi}_{\mathscr C^{4d}(\RealNum^{2d},\HilbertSpace)}
\)
and, for \(0 < \epsilon < 1\),
\( \norm{\eta_\epsilon}_{\mathscr C^{4d}(\RealNum^{2d})} 
\leq  \norm{\eta}_{\mathscr C^{4d}(\RealNum^{2d})} \)
 using the definition of \(\eta_\epsilon\) and the chain rule.
 Hence, by dominated convergence, it is sufficient to verify pointwise convergence
 of the integrand. To this end we write
  \[
      D_{\text{reg}}(\partial_x, \partial_y) \eta_\epsilon(x,y) \Psi(x,y)
      = \eta_\epsilon(x,y)    D_{\text{reg}}(\partial_x, \partial_y) \Psi(x,y) 
      + R_\epsilon(x,y). \numberthis
  \]
  In each term of the product rule expansion of the remainder~\(R_\epsilon\)
  there is at least one derivative with respect to \(x\) or \(y\) acting on
  \(\eta_\epsilon\). Thus this remainder is proportional to \(\epsilon\) and 
  vanishes for \(\epsilon \to 0\). The first term converges pointwise to
  \( D_{\text{reg}}(\partial_x, \partial_y) \Psi(x,y) \) and from this the
  claim~\eqref{eq:oscint} follows. 
  \qed
  \end{Lem}

For our scattering theoretic purposes it is useful that this method
can be readily extended to obtain similar integral representations by
iterating \Cref{lem:oscint}. In this manner we can further strengthen the
decay of the integrand for large \(x,y \in \RealNum^d\).

\begin{Lem}
 \label{lem:oscint2}
  Let \(M\in \NaturalNum\) and
  let \(\RealNum^{2d} \ni (x,y) \longmapsto \Psi(x,y) \in \HilbertSpace\) be a
  map with uniformly bounded derivatives for all multi-indices~\(\beta\in
  \NaturalNum_0^{2d}\) up to order
  \(\abs{\beta} \leq 4Md\).
  Then we have
  \begin{align*}
    \lim_{\epsilon \to 0}
    \frac{1}{(2\pi)^{d}} \int \DInt[\sPone]x \; \DInt[\sPone]y \; \eta(\epsilon x,
    \epsilon y) \Ee^{- \Ii x \cdot y} \Psi(x,y)
    =
    \frac{1}{(2\pi)^{d}} \int \DInt[\sPone]x \; \DInt[\sPone]y \; \Ee^{- \Ii x \cdot y}
    D_{\text{\rm reg}}(\partial_x, \partial_y)^M \Psi(x,y).  \label{eq:oscint2}\numberthis
  \end{align*}
  The integrand satisfies the integrable norm bounds
  \begin{align*}
    \norm{D_\text{\rm reg}(\partial_{x}, \partial_{y})^M \Phi(x,y)} \leq
    \prod_{j=0}^{d-1}
    \frac{C_M}{(1+x_j^2+y_j^2)^{2M}} \norm{\Phi}_{\mathscr C^{4Md}(\RealNum^{2d},\HilbertSpace)}.
    \label{eq:majest2}\numberthis
  \end{align*}
  \end{Lem}

  Next, we use the convergent integral representation to check that warped
  convolutions can be exchanged with the smearing operations used to define the
  creation-operator approximants in the deformed model.
  \begin{Lem} \label{lem:wcexch}
    Let \(\Psi \in \FiniteEnergyStates\), \(A \in \Alg^{0r}(\Ww)\),
    \(\chi \in \SchwartzSpace(\RealNum^d)\), and let \(f\) be a regular
    Klein-Gordon solution. Then
    \[
      B_Q\Psi := A_{Q_\Ww}(\chi)\Psi = (A(\chi))_{Q_\Ww}\Psi,
      \label{eq:chiwarp}
      \numberthis
    \]
    and
    \[
      B_{Q\T}(f)\Psi = (B_\T(f))_{Q_\Ww} \Psi.
      \label{eq:kgwarp}
      \numberthis
    \]
    \proof 
    We can apply \Cref{lem:oscint} to warped convolutions, as the corresponding
    integrand \(\Psi(x,y) = \alpha_{Qx}(A) U(y) \Psi\) is arbitrarily often
    differentiable for \(A \in \Cinfty\) and \(\Psi \in \FiniteEnergyStates\),
    with \(\norm{\partial_x^\alpha \partial_y^\beta \Psi(x,y)} \leq
      C_{\alpha\beta}\) for suitable constant depending on the multi-indices
      \(\alpha, \beta \in \NaturalNum_0^d\).
    By translation covariance of warped convolutions from \Cref{thm:wc} \itref{it:deformcov}
    we have
    \begin{align*}
      A_{Q_\Ww}(\chi) \Psi &= 
      \int \DInt[d]z \, \chi(z) \alpha_z(A_{Q_\Ww}) \Psi 
      =
      \int \DInt[d]z \, \chi(z) (\alpha_zA)_{Q_\Ww} \Psi 
      \\&= \frac 1 {(2\pi)^d} \;
    \int \DInt[d]z \, \chi(z)
    \int \DInt[d]x \, \DInt[d]y \;\Ee^{- \Ii x \cdot y}
    D_{\text{reg}}(\partial_x, \partial_y) U(x) \alpha_{Qy+z}(A)  \Psi.
    \numberthis
    \end{align*}
    Using the decay of \(\chi \in \SchwartzSpace(\RealNum^d)\) and 
    estimate \eqref{eq:majest}, the above integrand has integrable norm
    with respect to the product Lebesgue measure \(\DInt[3d](x,y,z)\). Hence the order of integrations
  can be exchanged by Fubini's theorem. Taking the \(x\)- and \(y\)-dependent
  translation operators outside the inner strong integral, we obtain
  \eqref{eq:chiwarp}. The proof of \eqref{eq:kgwarp} is analogous.
  \qed
  \end{Lem}

  \begin{Def}
    For
    \(
      \Psi, \Psi' \in \Gamma^u(\HilbertSpace_1)
    \)
    and \(Q \in \ComplexNum^{d \times d}\)
    let
    \[
      \Psi \otimes_Q \Psi' :=  \Ee^{\Ii P_1 \cdot Q P_2} \Psi \otimes \Psi'.
      \numberthis \label{eq:defTensor}
    \]
    Here, \(P_1 = P \otimes \Id\) is the energy-momentum operator acting on the 
    first argument only, and similarly \(P_2 = \Id
      \otimes P \) acts on the second argument.%
  \end{Def}

  Before we continue let us remark that this deformed tensor product preserves the
  ordered subspaces in the sense that if \(\Psi, \Psi' \in
    \Gamma^\succ(\HilbertSpace_1)\), with \(\Psi \succ \Psi'\), then also
    \( \Psi \otimes_Q \Psi' \in \Gamma^\succ(\HilbertSpace_1) \).
  We cautiously note that this deformed tensor product is clearly
  linear in its arguments and associative, but not commutative (unless \(Q =
  0\)).
  This is somewhat reminiscent of the general structure of the
  Zamolodchikov-Faddeev relations in integrable models, in contrast to canonical
  commutation relations (see e.g.~\cite{Le03}).
  Further, \(\otimes_Q\) is in general not
  mixed-associative in combination with ordinary tensor products. That is,
  \( \Psi_1 \otimes_Q (\Psi_2 \otimes \Psi_3)
    \not = (\Psi_1 \otimes_Q \Psi_2) \otimes \Psi_3 \).
  We also note that the definition is consistent with the fact that on Fock
  spaces \(\Psi \otimes \Omega = \Psi = \Omega \otimes \Psi\) are identified for
  any \(\Psi\).

  We can now establish the main lemma for proving the wave operator identity from \Cref{thm:defw}.
  \begin{Lem} \label{lem:deformconv} 
    Let \(\Psi_1^k = \Bb_{k\T}(f_k)\Omega = \Bb_{k\T}^\perp(f_k)\Omega\), \(1
    \leq k \leq n\) be
    swappable one-particle states, s.t.~\(A_k \in \Alg^{0r}(\Ww), A^\perp_k \in
      \Alg^{0r}(\Ww^\perp)\) are regular (\(1\leq k \leq n\)), and \(B_k =
    A_k(\chi)\) where \(\chi\in \SchwartzSpace(\RealNum^d)\) is an admissible
    Haag-Ruelle auxiliary function.  For ordered velocity supports~\(\VS_1
    \succ \ldots \succ \VS_n\) we have
   \[
     \lim_{\T\to\infty}
   \BB_{1Q\T}(f_1) \BB_{2Q\T}(f_2) \ldots \BB_{nQ\T}(f_n)
     \Omega
     =
        \WO^{+}_{0,\Ww} 
      \Psi_1^1 \otimes_{Q_\Ww} \ldots \otimes_{Q_\Ww} \Psi_1^n.
        \label{eq:defidout}
        \numberthis
   \]
   With incoming ordering~\(\VS_1 \prec \ldots \prec \VS_n\)  
   we have analogously for the incoming limit \(\T \to -\infty\)
 \[
     \lim_{\T\to-\infty}
   \BB_{1Q\T}(f_1) \BB_{2Q\T}(f_2) \ldots \BB_{nQ\T}(f_n)
     \Omega
     =
        \WO^{-}_{0,\Ww} 
      \Psi_1^1 \otimes_{Q_\Ww} \ldots \otimes_{Q_\Ww} \Psi_1^n.
        \label{eq:defidin}
        \numberthis
   \]
  \proof 
  We consider only the case \(\T\to\infty\). We first note that \(\Psi_{k\T}  := 
  \BB_{kQ\T}(f_k)  \ldots \BB_{nQ\T}(f_n) \Omega
    \in \FiniteEnergyStates\)
    for \(1 \leq k \leq n\),
  due to the compact energy-momentum transfer of the \(\BB_{jQ\T}(f_j)\), \(1
  \leq j \leq n\)
  (see \cite{MD17a} Lemma~7).
  Thus Lemmas~\ref{lem:wcexch} and \ref{lem:oscint} apply and we obtain
  \begin{align*}
    \Psi_{\T} &:= \Psi_{1\T} =  \Bb_{1\T}(f_1)_{Q_\Ww}
    \BB_{2\T}(f_2)_{Q_\Ww} \ldots \BB_{n\T}(f_n)_{Q_\Ww}\Omega
   \\&=
     \int \frac{\DInt[\sPone]x_1 \; \DInt[\sPone]y_1 \ldots
       \DInt[\sPone]x_{n} \; \DInt[\sPone]y_{n} 
   }{(2\pi)^{n \sPone}} \;
   \prod_{j=1}^{n}
   ( \Ee^{- \Ii x_j \cdot y_j}
   D_{\text{reg}}(\partial_{x_j}, \partial_{y_j})^M)
   \\&  \qquad\qquad
     U(x_1)\alpha_{Qy_1}(\Bb_{1\T}(f_1))
     \ldots
   U(x_{n})\alpha_{Qy_{n}}(\Bb_{n\T}(f_n)) \Omega,
     \, 
       \numberthis
       \label{eq:convrepn}
  \end{align*}
  where the constant \(M \in \NaturalNum\) will be chosen below.
  We rewrite the vector part of the integrand 
using that
 \(
     U(x_1)\alpha_{Qy_1}(\Bb_{1\T}(f_1))
     =
     \alpha_{Qy_1+x_1}(\Bb_{1\T}(f_1))
     U(x_1)
  \)
  as
  \begin{align*}
    \Psi_{\uvec x, \uvec y, \T} &:=
     U(x_1)\alpha_{Qy_1}(\Bb_{1\T}(f_1))
     \ldots
   U(x_{n})\alpha_{Qy_{n}}(\Bb_{n\T}(f_n)) \Omega
   \\&=
     \alpha_{Qy_1+x_1}(\Bb_{1\T}(f_1))
     \alpha_{Qy_2+x_1+x_2}(\Bb_{2\T}(f_2))
     \ldots
   \alpha_{Qy_{n} +x_1 +\ldots +x_n}(\Bb_{n\T}(f_n)) \Omega,
   \numberthis
  \end{align*}
  with \(\uvec x = (x_1, \ldots, x_n)\), \(\uvec y=(y_1, \ldots, y_n) \in
  \RealNum^{nd}\), and we abbreviate \(Q:=Q_\Ww\).

  Now we split the integration in \eqref{eq:convrepn} into an integral over the region
  \[R_{\rho\T}^\uparrow: = \{(x_1, \ldots, x_n,y_1, \ldots, y_n) \in
      \RealNum^{2nd} : |x_k^0| + |\vec x_k| \leq
  \rho \T, |y_k^0|+|\vec y_k| \leq \rho \T\}
  \numberthis 
\] and its complement, where the constant \(\rho > 0\) is for now
  fixed but arbitrary, and will be specified later. The vector resulting from the integral over
  \(R_{\rho\T}^\uparrow\) will be denoted \(\Psi_{1\T}^\uparrow\).
  By means of applying~\eqref{eq:majest2} iteratively for all pairs
  \((x_k,y_k)\), \(1\leq k \leq n\), we verify that with an appropriately strong
  power \(M \in \NaturalNum\) of the regularizing operator \(D_{\text{reg}}\),
  the remainder integral over the complement \(R_{\rho\T}^\downarrow : =
  \RealNum^{2nd}\setminus R_{\rho\T}^\uparrow\) becomes small for large~\(\T\).
  That is,
\begin{align*}
  \norm{\Psi_\T^\downarrow} &:= 
  \norm{\Psi_\T - \Psi_\T^\uparrow}
  \\&\leq
    \int_{R_{\rho\T}^\downarrow}
    \frac{\DInt[n\sPone]\uvec x \; \DInt[n\sPone]\uvec y}{(2\pi)^{n\sPone}} 
    \;
  \left\|
    \left(
    \prod_{k=1}^{n}
   D_{\text{reg}}(\partial_{x_k}, \partial_{y_k})^M
 \right)
   \Psi_{\uvec x, \uvec y, \T}
 \right\|
   \numberthis\label{eq:beforeiterativeest}
   \\&\leq
    \int_{R_{\rho\T}^\downarrow}
    \frac{\DInt[n\sPone]\uvec x \; \DInt[n\sPone]\uvec y}{(2\pi)^{n\sPone}} 
    \left(
   \prod_{k=1}^{n}
   \prod_{j=0}^{d-1}
 \frac{1}{(1+x_{k,j}^2 +y_{k,j}^2)^{2M}}\right) C   (1+|\T|^{ns/2}).
   \numberthis\label{eq:iterativeest}
\end{align*}
In the last step we estimated
  \begin{align*}
  \fl 
    \norm
    { D_{\text{reg}}(\partial_{x_1}, \partial_{y_1})^M
      \cdots D_{\text{reg}}(\partial_{x_n}, \partial_{y_n})^M
  \Psi_{\uvec x, \uvec y, \T} }
  \\
  &\leq C
 \prod_{k=1}^{n} \prod_{j=0}^{d-1} \frac{ 1  }{(1+x_{k,j}^2 +y_{k,j}^2)^{2M}} 
 \norm{\Psi_{\uvec x, \uvec y, \T} }_{{\mathscr C}^{4Mnd}
 (\RealNum^{2nd},\HilbertSpace)}.
 \numberthis \label{eq:regnormest}
  \end{align*}
The derivative norm of the vector part was then bounded by expanding it
into individual differentiated terms \(\partial_{\uvec x}^{\uvec \alpha}\partial_{\uvec
y}^{\uvec \beta} \Psi_{\uvec x, \uvec y, \T}\), which in turn can be expanded by
the product rule into terms with differential operators acting on the translated
creation-operator approximants \(B_{k\T}(f_k)\). Due to the assumed
norm differentiability of \(A_k\), the differentiated \(B_{k\T}(f_k)\) can be
rewritten as creation operator approximants constructed using
the differentiated~\(A_k\). These will be denoted here by \(\tilde B_{k\T}(f_k)\)
and \(\tilde A_k\), respectively. Using this rewriting, the terms from the expanded
differentiated vector parts from \eqref{eq:beforeiterativeest} can each be
bounded using the standard estimate \(\normm{\tilde B_{k\T}(f_k)} \leq C_{\chi,
  f_k} \normm{\tilde A_k} (1+\abs{\T}^{s/2})\)
to obtain
\(\normm{\alpha_{z_1}(\tilde B_{1\T}(f_1)) \ldots \alpha_{z_n}(\tilde
  B_{n\T}(f_n))\Omega} \leq \prod_{k=1}^n \normm{\tilde B_{k\T}(f_k)} \leq C
      (1+\abs{\T}^{ns/2})
\).  Taking all these terms together we obtain the last step from
\eqref{eq:iterativeest},
where the new constant \(C\) depends on~\(M\), on all wave packages and on the
norms of derivatives of the \(A_k\) operators up to order \(4Mnd\).
To obtain the decay estimate for sufficiently large \(\T\) we proceed to estimate
\begin{align*}
  \norm{\Psi_\T^\downarrow} & 
  \leq C \abs{\T}^{ns/2} \left( \sup_{(\uvec x, \uvec y) \in R_{\rho\T}^\downarrow}
   \prod_{k=1}^{n} \prod_{j=0}^{d-1} \frac{ 1 }{(1+x_{k,j}^2 +y_{k,j}^2)^{2M-2}} \right)
  \prod_{k=1}^{n} \prod_{j=0}^{d-1} \; \int\limits_{\RealNum^2}
   \frac{\DInt x_{k,j} \DInt y_{k,j}}{(1+x_{k,j}^2 +y_{k,j}^2)^{2}}.
   \numberthis
\end{align*}
Here constant factors such as the convergent integrals can be absorbed into the constant \(C\).
To bound the supremum we note that by definition of \(R_{\rho\T}^\uparrow\) we
have for any point \((\uvec x, \uvec y)\) in the complement
\(\RealNum^{2nd}\setminus R_{\rho\T}^\uparrow\)
at least one vector with \(|x_{k^*}^0|+ |\vec x_{k^*}| > \rho\T \) or
\(|y_{k^*}^0|+|\vec y_{k^*}|>\rho\T\).
This in turn implies that at least one coordinate \(|x_{k^*,j^*}|\) or
\(|y_{k^*,j^*}|\),
respectively, is larger than \( \rho\T/\sqrt{4s}\). As all other factors in the
supremum are bounded from above by one, we obtain
\begin{align*}
  \norm{\Psi_\T^\downarrow} &
 \leq C |\T|^{ns/2+ 4-4M}. \numberthis 
 \label{eq:coarsebound}
\end{align*}
Choosing \(M \in \NaturalNum\) large enough, this contribution becomes
arbitrarily small for large \(\T>0\).

In particular,
for establishing \eqref{eq:defidout}
it is sufficient to consider the limit of \(\Psi^\uparrow_\T\). To address this
convergence, let \(\tilde \rho> 0\) denote the minimum over all such constants from
\Cref{lem:unifbound} for the families of \((\tilde B_{j\T}(f_j))_{1\leq j
    \leq n}\), where \(\tilde B_{j\T}(f_j) = \partial_{x_j}^{\beta_j}
  \alpha_{x_j} (B_{j\T}(f_j))\big|_{x_j=0} \) stand for all possible combinations of
  derivatives of these operators up to order \(|\beta_j| \leq 4Mnd\), \(\beta_j
  \in\NaturalNum_0^d\). It is used here that 
  \( \partial_{x_j}^{\beta_j} \alpha_{x_j}(B_{j\T}(f_j)) \Omega =
    \partial_{x_j}^{\beta_j}
  \alpha_{x_j}(B_{j\T}^\perp(f_j)) \Omega \) for all multi-indices
  \(\beta_j\in\NaturalNum_0^d\),
  which shows by setting \(x_j=0\) that all such \(\tilde B_{j\T}(f_j)\Omega\) are
swappable, as needed to apply \Cref{lem:unifbound} to these differentiated families.
Further the derivatives with respect to \(y_k\) can also be written as
multiples of such differentiated operators by using the chain rule. Thus we set \(\rho :=
  \tilde \rho \cdot (1+\|Q\|)^{-1}/(n+1) < \infty \). With this choice we have
  for any \(1\leq k \leq n\) and \((\uvec x, \uvec
y) \in R^\uparrow_{\rho \T}\) that the vectors 
\(
  z_k := Qy_k + x_1+ \ldots +x_k
\)
are contained in the double cone of radius \(k\rho\abs \T + \|Q\| \rho \T \leq
\tilde \rho \abs{\T} \)
for \(1\leq k \leq n\).
Considering the asymptotically dominant part
\begin{align*}
 \Psi_{\T}^\uparrow &:=
  \int_{R_{\rho\T}^\uparrow} \frac{
       \DInt[n\sPone]\uvec x \; \DInt[n\sPone]\uvec y
   }{(2\pi)^{n \sPone}} \;
   \left(
   \prod_{j=1}^{n}
    \Ee^{- \Ii x_j \cdot y_j}
   D_{\text{reg}}(\partial_{x_j}, \partial_{y_j})^M\right)
  \Psi_{\uvec x, \uvec y, \T},
  \numberthis
  \end{align*}
  we can obtain a \(\T\)-uniform integrable bound by applying
  \eqref{eq:regnormest} and noting that instead of the coarse bound
  \eqref{eq:coarsebound} we can now estimate
  \begin{align*}
    \CharFct_{R_{\rho\T}^\uparrow}(\uvec x, \uvec y)
    \norm{\Psi_{\uvec x, \uvec y, \T} }_{{\mathscr
    C}^{4Mnd}(\RealNum^{2nd},\HilbertSpace)} \leq C.
    \numberthis
  \end{align*}
  Here we make use of \Cref{lem:unifbound}, noting that this uniform bound applies to each of the terms resulting
  from the expansion \[
    \norm{\Psi_{\uvec x, \uvec y, \T} }_{{\mathscr C}^{4Mnd} (\RealNum^{2nd},\HilbertSpace)}
    = \sum_{\substack{\uvec \alpha, \uvec \beta \in \NaturalNum_0^{nd},
        \\\abs{\uvec \alpha}+\abs{\uvec \beta}\leq 4Mnd}}\normm{\partial_{\uvec x}^{\uvec \alpha} \partial_{\uvec
  y}^{\uvec \beta}\Psi_{\uvec x, \uvec y, \T} }
  \numberthis
  \]
 by our choice of \(\rho\).
Due to the restriction to \(R_{\rho\T}^\uparrow\) we can estimate the 
differentiable vector norm by means of \Cref{lem:unifbound}, after expanding it
into individual terms with fixed derivatives acting on the
creation-operators~\(\tilde B_{k\T}(f_k)\). As before these can be written as Haag-Ruelle-type
operators in terms of the differentiated \(A_k\), so that \Cref{lem:unifbound}
applies. 
To summarize we note that the proof strategy here is in fact analogous to the
above method used for the outside region \(R_{\rho\T}^\downarrow\). However in
  \(R_{\rho\T}^\uparrow\) the use of the clustering bound of
  \Cref{lem:unifbound} is geometrically permitted and yields the much stronger
  \(\T\)-uniform estimate on the vector part.

  By dominated convergence we obtain
  \begin{align*}
    \lim_{\T\to\infty} \Psi_{\T} & \overset{\eqref{eq:coarsebound}}=
    \lim_{\T\to\infty} \Psi_{\T}^\uparrow=
\int \frac{
       \DInt[n\sPone]\uvec x \; \DInt[n\sPone]\uvec y
   }{(2\pi)^{n \sPone}} \;
   \lim_{\T\to\infty}
   \prod_{j=1}^{n}
   ( \Ee^{- \Ii x_j \cdot y_j}
   D_{\text{reg}}(\partial_{x_j}, \partial_{y_j})^M)
  \Psi_{\uvec x, \uvec y, \T},
  \numberthis\label{eq:limexch}
  \end{align*}
  where we already used that the characteristic function
  \( \CharFct_{R_{\rho\T}^\uparrow}(\uvec x, \uvec y) \to 1 \) 
  pointwise for \(\T \to \infty\).
  Here the Haag-Ruelle limit can be exchanged with the regularizing differential
  operators by explicit computation: we expand everything into differentiated
  translated Haag-Ruelle operators, for which the right hand side of
  \eqref{eq:limexch} converges to the scattering state generated by the
  differentiated and translated operators, which are again Haag-Ruelle type
  creation-operator approximants. Collecting the Haag-Ruelle limits again after
  performing them, we have by linearity of the wave operator
  \begin{align*}
    \fl[2em]
    \lim_{\T \to \infty}
    \left(
   \prod_{j=1}^{n}
   D_{\text{reg}}(\partial_{x_j}, \partial_{y_j})^M
 \right)
  \Psi_{\uvec x, \uvec y, \T}
  \numberthis
  \\& = \WO^+_{0,\Ww}
  \left(
   \prod_{j=1}^{n}
   D_{\text{reg}}(\partial_{x_j}, \partial_{y_j})^M 
 \right)
     (U(Qy_1+x_1)\Psi_1^1)
     \otimes 
     (U(Qy_2+x_1+x_2)\Psi_1^2)
     \otimes
     \ldots
   \\&
   \hspace{24em} \otimes
     ( U(Qy_{n} +x_1 +\ldots +x_n)\Psi_1^n)
     \\ &
 = \WO^+_{0,\Ww}
 \left(
   \prod_{j=1}^{n}
   D_{\text{reg}}(\partial_{x_j}, \partial_{y_j})^M 
 \right)
   U(x_1)
   \left\{ \bigg.
     (U(Qy_1)\Psi_1^1)
     \otimes 
     U(x_2)
       \left\{
         \big.
         (U(Qy_2)\Psi_1^2)
     \otimes
     \ldots
     \right.\right.
     \\&
     \left.\left.
         \left.
   \hspace{24em}
     \otimes
     U(x_n)\left\{( U(Qy_{n})\Psi_1^n)
 \right\}\right.\ldots\big. \right\}\bigg. \right\}.
  \end{align*}
  In the last equality we have written the one-particle state translations
  again in groups using the second quantized translations on the unordered Fock space.  In this form
  we can now apply the warped-convolution-type integral representation of the
  deformed tensor product (\Cref{lem:deftprep}).
  For this purpose we note that the wave operator is bounded on the range of the
  integrand and can therefore be taken outside the strong integral. Thus the
  introduction of the
  regularizing differential operators can be undone by iterative application of
  \Cref{lem:oscint2} and we obtain
  \begin{align*}
  \lim_{\T \to \infty} \Psi_\T &= 
    \WO^+_{0,\Ww} \left(
\lim_{\epsilon_1, \ldots, \epsilon_n \to 0}
\int \frac{ \DInt{\uvec x} \, \DInt{\uvec y}}{(2\pi)^{nd}}
\left(
\prod_{j=1}^n  \eta(\epsilon_j x_j, \epsilon_j y_j) \Ee^{-\Ii x_j \cdot y_j}
\right)
\right.
\\&\left.
  U(x_1) \left\{ (U(Qy_1)\Psi_1^1) \otimes U(x_2) \left\{(U(Qy_2)\Psi_1^2)
      \otimes \ldots \otimes U(x_n)\left\{( U(Qy_{n})\Psi_1^n)
        \bigg.
\right\} \right. \right\} \bigg. \right).
\numberthis
  \label{eq:finwoproof}
  \end{align*}
Regrouping the convergent integrals, using Fubini and continuity of the tensor
product, and applying \Cref{lem:deftprep} iteratively, we obtain deformed tensor
products, as claimed in~\eqref{eq:defidout}. We note that after replacing all
tensor products by \(Q\)-deformed tensor products the right-associative grouping
from \eqref{eq:finwoproof} becomes inessential and can be dropped.
The proof of the statement for the incoming limit \(\T\to-\infty\) is analogous.  \qed

 \end{Lem}

We used the following auxiliary result concerning the norm of scattering-state
approximants involving certain restricted translations of each operator.
\begin{Lem} \label{lem:unifbound}
  For any family  of
  swappable
  operators  \(B_{k\T}(f_k)\), \(1\leq k \leq n\), 
  with regular Klein-Gordon wave packets satisfying the outgoing ordering
  \(
    \VS_1 \succ \ldots \succ \VS_n
  \) 
  there exists a constant \(\rho > 0\) such that
  for all \(\T \geq 0\) and \(x_1, \ldots, x_n \in \DoubleCone_{\rho|\T|}\) 
  \begin{align*}
    \norm{\alpha_{x_1}(B_{1\T}(f_1)) \ldots \alpha_{x_n}(B_{n\T}(f_n)) \Omega } \leq
    C,
    \label{eq:normbound}\numberthis
  \end{align*}
  where  \(\DoubleCone_r := \{ x = (x^0, \vec x) \in \RealNum^d:
    \abs{x}_c:=\abs{x^0}+\abs{\vec x} < r\} = r \DoubleCone_1 
  \) denotes the double cone of
    radius \(r > 0\).
  With opposite ordering
  \( 
    \VS_1 \prec \ldots \prec \VS_n
  \)
  an analogous bound holds for incoming times \(\T < 0\).
\end{Lem}

  Before proving this lemma let us recall the following useful technical result 
  from \cite{MD17a}, which concerns the approximation of Haag-Ruelle
  creation-operator approximants by wedge-local operators.

  \begin{Lem}[\cite{MD17a} Lemma 9]
    \label{lem:wedgeLocalBT}
    Let \(\Aa \in \Alg(\Ww)\). 
    For any \(\T \in \RealNum\) and \(\delta > 0\) the corresponding 
    \(\BB_\T :=  \BB_{\T}(f)\) can be approximated 
    by \(\BB_{\T}^{(\delta)} \in \Alg(\T \VS_f + \DoubleCone_{\delta\abs{\T}} + \Ww) \),
  \((\delta > 0)\), such that
    for any \(N \in \NaturalNum\)
    \[
      \norm{\BB_{\T}^{(\delta)} - \BB_\T} \leq \frac{ C_N^\delta}{1+\abs{\T}^N},
      \label{eq:asymptDec}\numberthis
    \]
  where the constants \(C_N^\delta\) depend 
    on \(f\), \(\Aa\) and \(\chi\).
  \end{Lem}

  For the proof of \Cref{lem:unifbound} we use a corresponding version of the
  commutator estimate, which will be formulated as a separate lemma to be proven
  first. For this estimate certain translations of oppositely
  localized pairs are admitted, similarly to the corresponding translations
  appearing in \Cref{lem:unifbound}. 

  \begin{Lem}
  \label{lem:CommEst}
  Let \(B = A(\chi)\), \(B^\perp=A^\perp(\chi)\) 
  with \(A \in \Alg(\Ww)\), \(A^\perp \in \Alg(\Ww^\perp)\), for a pair of
    opposite wedges \(\Ww, \Ww^\perp\),  \(\chi \in
  \SchwartzSpace(\RealNum^d)\), and 
  let \(f, f^\perp\) be regular Klein-Gordon solutions
  ordered by \(\VS_{f^\perp} \prec_{\Ww} \VS_f\).  Then 
  there exists a constant \(\rho> 0\) and 
  for any \(N \in \NaturalNum\) a constant \(C_N> 0\) such that
  for any \(\T > 0\) and \(x,y  \in \DoubleCone_{\rho\T} \),
  \[
    \norm{\left[\alpha_y(\BB_\T^{\perp}(f^\perp)), \alpha_x(\BB_\T(f)) \right] }
    \leq \frac{C_N}{1+\abs{\T}^N}.
    \label{eq:commEst}\numberthis
  \]
   This rapid decay extends to the case that either or both operators are
   replaced by their adjoint. Further these estimates also hold for \(\T < 0\)
   given the opposite ordering~\(\VS_{f^\perp} \succ_{\Ww} \VS_{f}\).
  \end{Lem}
In the proof presented here we will focus on the arguments needed to
generalize the corresponding commutator estimate of Corollary~10 from
\cite{MD17a}  to the present statement.  More details can be found in
Appendix~A of \cite{MD17a}.  

 \begin{proof} 
     We only consider the outgoing case \(\T >0\) and note that it is sufficient
     to prove \eqref{eq:commEst} for all \(\T \geq \T_0\) with some fixed
     \(\T_0>0\). Let us begin by assuming for simplicity that \(y=0\).
     For \(\delta > 0\) we obtain families of wedge-local
     approximants \(\BB_\T^{(\delta)}\), \( \BB_\T^{\perp(\delta)}\) to the
     respective Haag-Ruelle operators via \Cref{lem:wedgeLocalBT}. These satisfy
     \(\BB_\T^{(\delta)} \in \Alg(\T \VS_f + \DoubleCone_{\delta\abs{\T}} +
         \Ww)\), with \( \normm{\BB_{\T}^{(\delta)} - \BB_\T(f)} \leq {
      C_M^\delta}/(1+\abs{\T}^M)\), and analogously
      \(\BB_\T^{\perp(\delta)} \in \Alg(\T \VS_{f^\perp} +
        \DoubleCone_{\delta\abs{\T}} + \Ww^\perp)\), with
    \(
      \normm{\BB_{\T}^{\perp(\delta)} - \BB_\T^\perp(f^\perp)} \leq
    { C'^{\delta}_M}/(1+\abs{\T}^M)\) with constants \(C_M^\delta\),
    \(C_M'^\delta > 0\), provided for all \(M \in   \NaturalNum\) by \Cref{lem:wedgeLocalBT}.
    
   Proceeding towards \eqref{eq:commEst} we obtain from translation covariance
   that  \(\normm{\alpha_x(\BB_\T^{(\delta)}) - \alpha_x(\BB_\T(f)) } 
      \leq {C'^\delta_M}/(1+\abs{\T}^M) \) 
   and \(\alpha_x(\BB_\T^{(\delta)}) \in 
      \Alg(\T \VS_f + \DoubleCone_{\delta\abs{\T}} + \Ww + x)\).
      To obtain \eqref{eq:commEst} by means of the wedge-locality  of these
      two approximating operators we have to choose \(\delta > 0\) and
      subsequently \(\rho>0\) sufficiently small so that the localization
      regions are space-like separated. These causality considerations
      can be simplified for large enough \(\T\) by absorbing any finite
      translations into the growing double cones and rewriting 
      as a region involving a single growing double cone:
      \(\DoubleCone_{\delta\T} + x \subset \DoubleCone_{\delta\T}
      +\DoubleCone_{\rho\T} \subset \DoubleCone_{(\delta+\rho){\T}}\).
     Further we can write \(\Ww = \Ww_c + x_\Ww\), \(\Ww^\perp = \Ww_c' +
     x_{\Ww^\perp}\) for some \(x_\Ww, \, x_{\Ww^\perp} \in \RealNum^d\).
     Choosing for simplicity \(\rho = \delta\) and assuming \(\T >
       \T_0:=
     (\abs{x_\Ww}_c + \abs{x_{\Ww^\perp}}_c)/\delta\) we obtain
     \begin{align*}
      M_1^\T := 
       \T \VS_f + \DoubleCone_{\delta\T} + \Ww + x
       & \subset \T \VS_f + \DoubleCone_{3\delta\T} + \Ww_c 
       , \; \text{and}
       \\
       M_2^\T :=
\T \VS_{f^\perp} +
        \DoubleCone_{\delta\T} + \Ww^\perp
        &
        \subset
\T \VS_{f^\perp} +
        \DoubleCone_{3\delta\T} + \Ww'_c.
        \numberthis \label{eq:regions}
     \end{align*}
     By the ordering assumption we have that \(\VS_f - \VS_{f^\perp}\)
     is a compact subset of the open set~\(\Ww_c\). In particular there exists an
     \(\epsilon > 0\) with \( \VS_f - \VS_{f^\perp} + \DoubleCone_\epsilon
     \subset \Ww_c\). Then for \(\T > 0\) also
     \(\T\VS_f - \T\VS_{f^\perp} + \DoubleCone_{\epsilon\T}
     \subset \Ww_c\). Thus we can choose \(\rho = \delta := \epsilon/6\) and it then
     follows that the two sets \(M_1^\T\) and 
     \(M_2^\T\) are space-like separated for \(\T > \T_0\).

     We now obtain from locality that for all \(\T \geq \T_0\) and \(x \in
    \DoubleCone_{\rho\T}\) we have \([ \BB_\T^{\perp(\delta)}, \alpha_x(\BB_\T^{(\delta)})] = 0\),
      which implies the uniform commutator estimate by
    expanding
    \begin{align*}
      \norm{[\BB^\perp_\T(f^\perp), \alpha_x(\BB_\T(f))]}
    & = 
    \norm{[\BB^\perp_\T(f^\perp) - \BB^{\perp(\delta)}_\T + \BB^{\perp(\delta)}_\T
    , \alpha_x(\BB_\T(f) - \BB^{(\delta)}_\T + \BB^{(\delta)}_\T)  ]}
    \\&
    \leq
    \norm{[\BB^\perp_\T(f^\perp) - \BB^{\perp(\delta)}_\T,
    \alpha_x(\BB_\T(f) - \BB^{(\delta)}_\T + \BB^{(\delta)}_\T)  ]}
    \\& \qquad\qquad
    + \norm{[ \BB^{\perp(\delta)}_\T  , \alpha_x(\BB_\T(f) - \BB^{(\delta)}_\T)   ]}
    + \norm{[ \BB^{\perp(\delta)}_\T  ,  \alpha_x(\BB^{(\delta)}_\T) ]},
    \label{eq:proofcomm} \numberthis
    \end{align*}
    where \(\normm{[\BB^\perp_\T(f^\perp) -
        \BB^{\perp(\delta)}_\T, \alpha_x(\BB_\T(f))  ]}  \leq 2 C_{N'}^\delta
        C/(1+\abs{\T}^{N'}) \cdot (1+\abs{\T})^{s/2} \leq C_N' \T^{-N} \)
        by estimating the commutator via the two operator norms, using \Cref{lem:wedgeLocalBT} for
        its first and the standard polynomially growing norm estimate for its second argument, and
  analogously for the other non-vanishing commutator.

  Finally, if \(y \not = 0\) we can write by translation-covariance of the
  operator norm that
     \[
    \norm{\left[\alpha_y(\BB_\T^{\perp}(f^\perp)), \alpha_x(\BB_\T(f)) \right] }
    = \norm{\left[\BB_\T^{\perp}(f^\perp), \alpha_{x-y}(\BB_\T(f)) \right] }
   \leq 
     C_N/(1+|\T|^N),
     \label{eq:exty}\numberthis
     \]
     for all \(x-y \in \DoubleCone_{\rho\T}\). We can thus set \(\tilde \rho:=
     \rho/2\) to obtain \eqref{eq:exty} for all \(x,y \in
     \DoubleCone_{\tilde\rho\T} \), using that then \(x-y \in
       \DoubleCone_{\rho \T/2}
     -  \DoubleCone_{\rho \T/2}
     \subset 
       \DoubleCone_{\rho \T}
     \). This establishes \eqref{eq:commEst}. 

     Concerning the extensions to adjoints and \(\T < 0\) we note that the
     commutator estimates involving adjoints follow using the same approximation argument,
     noting that \eqref{lem:wedgeLocalBT} also directly yields wedge-local
     approximants of adjoint operators by the \(\CStar\)-property of the
     operator norm. For \(\T < 0\) the geometric situation in \eqref{eq:regions}
     is inverted, but the same arguments work if the opposite
     ordering \(\VS_{f^\perp} \succ_\Ww \VS_f\) holds.
    \end{proof}

  \begin{proof}[Proof of \Cref{lem:unifbound}]
    For the outgoing case, let \(\rho = \rho_n > 0\) be the minimum over the
    constants~\(\rho_{k,j}\)  from \Cref{lem:CommEst}   for the pairs
    \(B_{k\T}(f_k)\) and \(B_{j\T}^\perp(f_j)\), \(1 \leq k < j \leq n\).

 For \(n=1\), estimate \eqref{eq:normbound} follows directly from the fact that
 \(\norm{\Psi_1^{x_1}} := \norm{\alpha_{x_1}(B_\T(f))\Omega} =
 \norm{U(x_1)B_\T(f)\Omega}\) does not depend on the translation vector~\(x_1\)
 due to translation invariance of the norm and the vacuum, neither on \(\T\) by
  construction of the Haag-Ruelle operators. 
  
  For the case \(n  \geq 2\),
  let \(\Psi_n^{\uvec x_n}(\T) := \alpha_{x_1}(B_{1\T}(f_1)) \ldots
    \alpha_{x_n}(B_{n\T}(f_n)) \Omega \) with \(\uvec x_n := (x_1, \ldots, x_n)\in
  \RealNum^{nd}\). To simplify the notation we will drop the obvious wave packet dependences
  and write \(B_{k\T}^{x_k} := \alpha_{x_k}(B_{k\T}(f_k)) \), so that
  \(\Psi_n^{\uvec x_n}(\T) = B_{1\T}^{x_1} \ldots B_{n\T}^{x_n} \Omega \) and
  for later use \(B_{k\T}^{\perp x_k} := \alpha_{x_k}(B_{k\T}^\perp(f_k)) \). To
  give an
  inductive argument we can write using the swapping property
  \begin{align*}
    \norm{\Psi_n^{\uvec x_n}(\T)}^2 &=
      \left\langle\Omega,
        \BB_{n\T}^{x_n*}\ldots \BB_{1\T}^{x_1*} \,
        \BB_{1\T}^{x_1} \ldots \BB_{n\T}^{x_n} \Omega
      \right\rangle
      \\&=
      \left\langle\Omega,
        \BB_{n\T}^{x_n*}\ldots \BB_{1\T}^{x_1*} \,
        \BB_{1\T}^{x_1} \ldots \BB_{n-1,\T}^{x_{n-1}} \BB_{n\T}^{\perp x_n} \Omega
      \right\rangle
\\&=
      \left\langle\Omega,
        \BB_{n\T}^{x_n*} \BB_{n\T}^{\perp x_n}
        \BB_{n-1,\T}^{x_{n-1}*} 
        \ldots \BB_{1\T}^{x_1*} \,
        \BB_{1\T}^{x_1} \ldots 
        \BB_{n-1,\T}^{x_{n-1}} \Omega
        \right\rangle
        \\&\qquad\qquad
        +
    \left\langle\Omega,
      \BB_{n\T}^{x_n*} \left[ 
        \BB_{n-1,\T}^{x_{n-1}*} 
        \ldots \BB_{1\T}^{x_1*} \,
        \BB_{1\T}^{x_1} \ldots 
      \BB_{n-1,\T}^{x_{n-1}}, \BB_{n\T}^{\perp x_n} \right] \Omega
      \right\rangle.
      \numberthis\label{eq:normclucomm}
  \end{align*}
  The first term can be bounded using the clustering property of the Haag-Ruelle
  operators (\cite{MD17a} Prop.~8 (vi)), which yields
\begin{align*}
  \fl
 \left\langle\Omega,
        \BB_{n\T}^{x_n*} \BB_{n\T}^{\perp x_n}
        \BB_{n-1,\T}^{x_{n-1}*} 
        \ldots \BB_{1\T}^{x_1*} \,
        \BB_{1\T}^{x_1} \ldots 
        \BB_{n-1,\T}^{x_{n-1}} \Omega
        \right\rangle
   \\&=
\left\langle\Omega,
        \BB_{n\T}^{x_n*} \BB_{n\T}^{\perp x_n}
        \Omega \right\rangle \left\langle \Omega,
             \BB_{n-1,\T}^{x_{n-1}*} 
        \ldots \BB_{1\T}^{x_1*} \,
        \BB_{1\T}^{x_1} \ldots 
        \BB_{n-1,\T}^{x_{n-1}} \Omega
      \right\rangle 
      \\&= \norm{\BB_{n\T}^{x_n}\Omega}^2 \normm{\Psi_{n-1}^{\uvec x_{n-1}}(\T)}^2 \leq C.
      \numberthis
  \end{align*}
  Here the estimate is obtained uniformly for all 
   \(x_1, \ldots, x_{n} \in
  \DoubleCone_{\rho_{n} \T}\) 
  by using that 
  \(x_1, \ldots, x_{n-1} \in
  \DoubleCone_{\rho_{n} \T}
  \subseteq
\DoubleCone_{\rho_{n-1} \T}\) 
due to the induction hypothesis, and
  bounding the first factor as for \(n=1\).
  Note that here we use \(\rho_n \leq \rho_{n-1}\),   
  which follows directly by the definition \[\rho_n = \min_{ 1 \leq k < j \leq n} \rho_{k,j}
    \leq  \min_{ 1 \leq k < j \leq n-1} \rho_{k,j} = \rho_{n-1}, \numberthis \]
  where \(\rho_{k,j}\) denote the constants from \Cref{lem:CommEst}   for the pairs
    \(B_{k\T}(f_k)\) and \(B_{j\T}^\perp(f_j)\), \(1 \leq k < j \leq n\).

  Finally, the commutator term from \eqref{eq:normclucomm} is also seen to be
  bounded in \(\T\) (in fact, rapidly decreasing), uniformly for
     \(x_1, \ldots, x_{n} \in \DoubleCone_{\rho_{n} \T}\), by means of
  \Cref{lem:CommEst}: we expand the big commutator 
  in \eqref{eq:normclucomm} into a sum of terms of vacuum expectation values of
  operators of the form
  \begin{align*}
        \BB_{n\T}^{x_{n}*} 
        \ldots \BB_{1\T}^{x_1*} \,
        \BB_{1\T}^{x_1} \ldots &\left[ \BB_{k\T}^{x_k} , \BB_{n\T}^{\perp x_n} \right]
        \ldots
        \BB_{n-1,\T}^{x_{n-1}}, \quad\text{or}
        \\
    \BB_{n\T}^{x_{n}*} 
    \ldots & \left[ \BB_{k\T}^{x_k*} , \BB_{n\T}^{\perp x_n} \right] 
        \ldots \BB_{1\T}^{x_1*} \,
        \BB_{1\T}^{x_1} \ldots 
        \BB_{n-1,\T}^{x_{n-1}}.
        \numberthis
      \label{eq:clucommdecay}
  \end{align*}
  Each of those terms is now estimated using that
  \( \norm{ \left[\BB_{k\T}^{x_k*} , \BB_{n\T}^{\perp x_n} \right]}\) 
  and 
  \( \norm{ \left[\BB_{k\T}^{x_k} , \BB_{n\T}^{\perp x_n} \right]}\) 
  can for \(1 \leq k < n\) and \(x_k, x_n \in 
    \DoubleCone_{\rho_{n}\T} \) by definition of \(\rho_n\) be bounded by \( 
  C_M (1+\T)^{-M}\).
  Choosing \(M\in \NaturalNum\) sufficiently large we absorb
  the growth of the simple norm estimates 
  \( \normm{\BB_{j\T}^{x_j*}} = \normm{\BB_{j\T}^{x_j}}  \leq C_j
(1+\abs{\T}^{s/2})\) used for the remaining operators with \(1\leq j \leq n\)
in the respective terms of the expansion.
Together the desired uniform bound on \eqref{eq:normclucomm} is obtained.
  \end{proof}

\begin{Prop} 
  Let \(\Psi_1, \Psi_2 \in \FS^u(\HilbertSpace_1)\) be vectors of bounded
  energy-momentum. Then for any warping matrix \(Q \in \RealNum^{d\times
  d}\) their \(Q\)-deformed tensor product
  has the oscillatory integral representation
        \[
        \Psi_1 \otimes_Q \Psi_2 =
     \lim_{\epsilon\to 0}
       \int \frac{\DInt[\sPone]x \; \DInt[\sPone]y}{(2\pi)^{\sPone}}
\eta(\epsilon x, \epsilon y)
    \Ee^{-\Ii x \cdot y}
    U(x)\{(U(Qy) \Psi_1) \otimes \Psi_2\},
    \label{eq:oscintvec}\numberthis
      \]
      where  \(\eta \in \SchwartzSpace(\RealNum^{\sPone} \times \RealNum^{\sPone})\)
      such that
      \(\eta(0,0) = 1\).
    \proof
    First     we can infer from \Cref{lem:oscint}
that the limit in \eqref{eq:oscintvec} exists and is
    independent of \(\eta\) within the specified restrictions.
    We rewrite the expression under the limit on the right-hand side 
    via  the spectral calculus of the energy-momentum operators as
    \begin{align*}
      \fl
       \int \frac{\DInt[\sPone]x \; \DInt[\sPone]y}{(2\pi)^{\sPone}}
\eta(\epsilon x, \epsilon y)
    \Ee^{-\Ii x \cdot y}  \;
    \int
\DInt E(p) \;
\Ee^{ \Ii p \cdot x }
\left(\left(\int \DInt E(q)\; \Ee^{\Ii q \cdot Qy} \Psi_1  \right)\otimes  \;
\Psi_2\right) 
\\&=
        \int
       \DInt E_{P}(p) \DInt E_{P_1}(q) \;
       \left(
      \Psi_1 \otimes \Psi_2 \cdot
        \int \frac{\DInt[\sPone]x \; \DInt[\sPone]y}{(2\pi)^{\sPone}}
        \eta(\epsilon x, \epsilon y)
        \Ee^{- \Ii x \cdot y + \Ii p \cdot x + \Ii q \cdot Qy}
      \right).
        \numberthis \label{eq:rhsrewrite}
    \end{align*}
    Here the Fubini theorem for exchanging the order of integrations applies,
    where integrability with respect to the product measure follows from
    the bounded energy-momentum of \(\Psi_1, \Psi_2\) and the rapid decay of
    \(\eta\).
    Here \(P\) denote the energy-momentum operators on the full Fock space and
    \(P_1( \Psi_1 \otimes \Psi_2)  := (P \Psi_1) \otimes \Psi_2\) acts only on
    the second component.
    Let us denote the inner scalar integral from \eqref{eq:rhsrewrite} by
    \(I_\epsilon(p,q)\).
    By the uniqueness result of \Cref{lem:oscint}
    we can proceed to concretely choose
    \(\eta(x,y) := \Ee^{-\abs{x}_e^2 - \abs{y}_e^2}\),
    where \(\abs{\cdot}_e\) denote the Euclidean norm.
    Then an elementary calculation (see \Cref{prop:gaussianint}) shows that
    \(
      I_\epsilon(p,q)
        \longrightarrow \Ee^{-\Ii p \cdot Q q}
    \)
    pointwise for \(\epsilon \to 0\), and in addition \(I_\epsilon(p,q)\) is
    bounded  uniformly in \(p, q\)  for small \(\epsilon > 0\).
    By dominated convergence we obtain that \eqref{eq:rhsrewrite} yields for
    \(\epsilon \to 0\)
    \begin{align*}
  \int
       \DInt E_{P}(p) \, \DInt E_{P_1}(q) \;
      \Psi_1 \otimes \Psi_2 \cdot
      \Ee^{-\Ii p \cdot Qq} &=
      \Ee^{-\Ii P \cdot Q P_1}
      \Psi_1 \otimes \Psi_2 
      =
      \Ee^{-\Ii P_2 \cdot Q P_1}
      \Psi_1 \otimes \Psi_2 
      \\&
      =
      \Ee^{\Ii P_1 \cdot Q P_2}
      \Psi_1 \otimes \Psi_2 
      =
      \Psi_1 \otimes_Q \Psi_2,
      \numberthis
    \end{align*}
    where we used that \(P = P_1 + P_2\) and \(P_1 \cdot Q P_1 = 0\).
   \qed
  \label{lem:deftprep}
\end{Prop}

\begin{Prop}
  \label{prop:gaussianint}
  For any \(p, q \in \RealNum^d\) and any warping matrix \(Q \in
  \RealNum^{d\times d}\) we have
  \begin{align*}
    \lim_{\epsilon \to 0}
        \int \frac{\DInt[\sPone]x \; \DInt[\sPone]y}{(2\pi)^{\sPone}}
        \Ee^{- \abs{ \epsilon x}_e^2 - \abs{ \epsilon y}_e^2}
        \Ee^{- \Ii x \cdot y + \Ii p \cdot x + \Ii q \cdot Qy}
        = \Ee^{-\Ii p \cdot Qq}. \numberthis
  \end{align*}
  \proof
  We let \(\varepsilon := \epsilon^2\) and use anti-symmetry of \(Q\) with
  respect to the Minkowski scalar product to 
  write
  \begin{align*}
        \int \frac{\DInt[\sPone]x \; \DInt[\sPone]y}{(2\pi)^{\sPone}}
        \Ee^{-\varepsilon \abs{x}_e^2 - \varepsilon \abs{y}_e^2}
        \;
        \Ee^{- \Ii x \cdot y + \Ii p \cdot x + \Ii q \cdot Qy}
        &
        =
        \int \frac{\DInt[\sPone]x \; \DInt[\sPone]y}{(2\pi)^{\sPone}}
        \Ee^{-\varepsilon \abs{x}_e^2 - \varepsilon \abs{y}_e^2}
        \;
        \Ee^{- \Ii x \cdot y + \Ii p \cdot x - \Ii (Qq) \cdot y}
        \\&
        =: J_\varepsilon^d(p, Qq). \numberthis
  \end{align*}
  Let us express the Minkowski products in terms of ordinary scalar products
  involving the Minkowski metric \(g\). Then we can calculate
  component-wise by Fubini,
  \begin{align*}
         J_\varepsilon^d(p, p')
         &=
        \int \frac{\DInt[\sPone]x \; \DInt[\sPone]y}{(2\pi)^{\sPone}}
        \Ee^{-\varepsilon \abs{x}_e^2 - \varepsilon \abs{y}_e^2}
        \;
        \Ee^{- \Ii x^T g y +\Ii   p^T g x -\Ii  p'^T g y}
        = \prod_{\mu=0}^{d-1}
        J_\varepsilon^1(p^\mu, p'_\mu),
        \numberthis
      \end{align*}
      with
      \begin{align*}
         J_\varepsilon^1(p, p'):=
        \int \frac{\DInt x \; \DInt y}{2\pi}
        \Ee^{-\varepsilon x^2 - \varepsilon y^2}
        \Ee^{- \Ii (x y +  p  x - p' y)}.
        \numberthis
  \end{align*}
  Here we have substituted \(x = x_\mu' := g_{\mu\nu}x^\nu\).
  By elementary calculation one obtains
  \[
         J_\varepsilon^1(p, p') =
         \frac{1 }{\sqrt{1+4\varepsilon^2}}
         \Ee^{\frac{ -\Ii p p' - \varepsilon(p^2+
         p'^2)}{1+4\varepsilon^2}}
         \overset{\varepsilon \searrow 0} 
         \to \Ee^{-\Ii p p'}
         \numberthis
  \]
  from which the claim follows. \qed
\end{Prop}
\begin{proof}[Proof of \Cref{thm:defw}]
  We consider the outgoing case for a fixed wedge \(\Ww\). By definition  the linear combinations of ordered
  product states
  \[
    \Psi_n =
    \Psi_1^1 \otimes \ldots \otimes \Psi_1^n,
    \quad \Psi_1^1 \succ_\Ww \ldots \succ_\Ww\Psi_1^n, \quad (n\in \NaturalNum)
    \numberthis
  \]
  are dense in \(\FS^{\succ_\Ww}(\HilbertSpace_1)\). Further by
  Proposition~\ref{prop:densereg} such states
  can be approximated with arbitrarily small error by vectors
  \begin{align*}
    \tilde \Psi_n = 
    \tilde \Psi_1^1 \otimes \ldots \otimes \tilde \Psi^n_1 &= 
    B_{1\T}(f_1)\Omega \otimes \ldots \otimes B_{n\T}(f_n)\Omega 
    \\&= 
    B_{1Q\T}(f_1)\Omega \otimes \ldots \otimes B_{nQ\T}(f_n)\Omega,
    \quad \VS_{f_1} \succ_\Ww \ldots \succ_\Ww \VS_{f_n},
    \numberthis
    \label{eq:vecsmooth}
  \end{align*}
  generated by swappable Haag-Ruelle approximants \(B_{k\T}(f_k)\Omega =
    B_{k\T}^\perp(f_k)\Omega =\tilde \Psi_1^k \in \HilbertSpace_{1c}^\Ww \),
    which are constructed from
  regular \(A_k \in \Alg^{0r}(\Ww)\), \(A_k^\perp \in \Alg^{0r}(\Ww^\perp)\).
  Here the further restriction from swappable states to swappable states
  generated by regular operators is used to assure that the warped convolutions
  are well defined and we have just seen that the linear combinations of the
  vectors~\eqref{eq:vecsmooth} are also dense in
  \(\FS^{\succ_\Ww}(\HilbertSpace_1)\) even with this additional smoothness
    requirement.  By \Cref{thm:wc}~\itref{it:thmwcvac} it is clear that they
    yield the same one-particle vectors \(B_{kQ\T}(f_k)\Omega =
    B_{kQ\T}^\perp(f_k)\Omega = \tilde\Psi_1^k \) for all \(1\leq k \leq n\).
    For simplicity of notation we will be dropping the tilde for the remainder
    of the proof.

    By definition of the wave operator we now obtain
    \[
      \WO_{Q,\Ww}^+
      \Psi_n = \lim_{\T\to\infty}
    B_{1Q\T}(f_1)\ldots B_{nQ\T}(f_n)\Omega
    = \WO_{0,\Ww}^+ \Psi_1^1 \otimes_Q \ldots \otimes_Q \Psi_1^n,
  \numberthis
    \]
    where the last equality follows from \Cref{lem:deformconv}.
    By induction we obtain that
    \begin{align*}
      \Psi_1^1 \otimes_Q \ldots \otimes_Q \Psi_1^n
      &=
      \Psi_1^1 \otimes_Q (\Psi_1^2 \otimes_Q \ldots \otimes_Q \Psi_1^n)
      = \Psi_1^1 \otimes_Q S_{Q_\Ww}(\Psi_1^2 \otimes \ldots \otimes \Psi_1^n)
      \\ &= \Ee^{\Ii P_1 \cdot Q_\Ww (P_2+ \ldots + P_n)} \Psi_1^1 \otimes S_{Q_\Ww}(\Psi_1^2 \otimes \ldots \otimes \Psi_1^n)
      \\ &=  S_{Q_\Ww}( \Psi_1^1 \otimes  \Psi_1^2 \otimes \ldots \otimes
    \Psi_1^n).
  \numberthis
    \end{align*}
    Hence the claimed identity holds on a dense subspace
    and thus by continuity of the wave operators and \(S_Q\) on the full domain
    \(\FS^{\succ_\Ww}(\HilbertSpace_1)\)  of \(\WO_{Q,\Ww}^+\).
    The argument for \(\WO_{Q,\Ww}^-\) is analogous.
\end{proof}

\subsection{Scattering Data and Wedge-Transition Matrix Elements}
\label{sec:ScattDat}
The two-particle \(S\)-matrix of higher-dimensional deformed
or GL-type models was worked out in \cite{GL07,BS08} for any fixed wedge \(\Ww\).
Even if the initial and thus also the deformed model were Poincaré
covariant, these authors observed that the two-particle \(S\)-matrix is 
in fact not fully Lorentz covariant in higher dimensions \(d > 1+1\).
In \cite{MD17a} it was proposed to make this observation more precise by 
specifying the dependence of wave operators and all related asymptotic data
on the localization wedge~\(\Ww\) of the Haag-Ruelle operators explicitly. 
Due to translation covariance of the \(S\)-matrix this reduces to a dependence 
on a wedge modulo translations, or equivalently, a dependence on a centered
wedge.

\begin{Def}
  Let \(\Ww_\fin, \Ww_\ini\) be centered wedges.
  Following \cite{MD17a} we define \(S\)-matrices of the initial and deformed wQFT model,
  respectively, as maps between the incoming ordered Fock space
  \(\Gamma^{\prec_{\Ww_\ini}}(\HilbertSpace_1)\)
  and the outgoing space \(\Gamma^{\succ_{\Ww_\fin}}(\HilbertSpace_1)\) 
  by
 \begin{align*}
   S_{0, \fin\,\ini}^{\Ww_\fin \Ww_\ini} &:= (\WO_{0,\Ww_\fin}^+)^*
   \WO_{0,\Ww_\ini}^-\;,
                                          &
   S_{Q, \fin\,\ini}^{\Ww_\fin \Ww_\ini}& := (\WO_{Q,\Ww_\fin}^+)^*
   \WO_{Q,\Ww_\ini}^-\;.
   \numberthis
  \end{align*}
  Similarly we define wedge-transition maps between two final or two initial
  states, respectively, as
  \begin{align*}
    S_{0,\fin\,\fin}^{\Ww_2 \Ww_1} &:= (\WO_{0,\Ww_2}^+)^* \WO_{0,\Ww_1}^+\;,
                               &
    S_{Q,\fin\,\fin}^{\Ww_2 \Ww_1} &:= (\WO_{Q,\Ww_2}^+)^* \WO_{Q,\Ww_1}^+\;,
    \\
    S_{0,\ini\,\ini}^{\Ww_2 \Ww_1}&:= (\WO_{0,\Ww_2}^-)^* \WO_{0,\Ww_1}^-\;,
                              &
    S_{Q,\ini\,\ini}^{\Ww_2 \Ww_1}&:= (\WO_{Q,\Ww_2}^-)^* \WO_{Q,\Ww_1}^-\;,
    \label{eq:defsmatrix}\numberthis
  \end{align*}
  for any two centered wedges~\(\Ww_1, \Ww_2\). 
\end{Def}

As a direct consequence of \Cref{thm:defw} we obtain similar expressions
for the \(S\)-matrices and wedge-transition matrix elements in deformed
wedge-local models.
\begin{Cor}
  Let \( \Ww_\fin, \Ww_\ini, \Ww_1, \Ww_2\) be arbitrary centered wedges. The
  \(S\)-matrices and wave operators of a BLS-deformed wQFT model can be
  expressed in terms of the corresponding objects of the undeformed model:
  \label{cor:scattdef}
  \begin{enumerate}[(i)]
\item \label{it:sd}
\(
  S_{Q,\fin\,\ini}^{\Ww_\fin\Ww_\ini}
  = (S_{Q_{\Ww_\fin}}^{\succ_{\Ww_\fin}})^* S_{0,\fin\,\ini}^{\Ww_\fin\Ww_\ini}
  S_{Q_{\Ww_\ini}}^{\prec_{\Ww_\ini}}
\),
\item \label{it:wt} 
\(
  S_{Q,\fin\,\fin}^{\Ww_2\Ww_1}
  = (S_{Q_{\Ww_2}}^{\succ_{\Ww_2}})^*
  S_{0,\fin\,\fin}^{\Ww_2\Ww_1} S_{Q_{\Ww_1}}^{\succ_{\Ww_1}}
\), and
\(
  S_{Q,\ini\,\ini}^{\Ww_2\Ww_1}
  = (S_{Q_{\Ww_2}}^{\prec_{\Ww_2}})^*
  S_{0,\ini\,\ini}^{\Ww_2\Ww_1} S_{Q_{\Ww_1}}^{\prec_{\Ww_1}}
\).
\end{enumerate}
\proof 
We obtain the expression for the \(S\)-matrix \itref{it:sd} from the short
computation
\begin{align*}
  S_{Q,\fin\,\ini}^{\Ww_f\Ww_i}
  &= (\WO^{+}_{Q,\Ww_\fin})^* \WO^{-}_{Q,\Ww_\ini}
  = (\WO^{+}_{0,\Ww_\fin} S_Q^{\prec_{\Ww_\fin}})^* (\WO^{-}_{0,\Ww_\ini} S_Q^{\succ_{\Ww_\ini}})
  = (S_Q^{\prec_{\Ww_\fin}})^* (\WO^{+}_{0,\Ww_\fin})^* \WO_{0,\Ww_\ini}^{-} S_Q^{\succ_{\Ww_\ini}}
  \\&= (S_Q^{\prec_{\Ww_\fin}})^* S_{0,\fin\,\ini}^{\Ww_\fin\Ww_\ini} S_Q^{\succ_{\Ww_\ini}}.
  \numberthis
\end{align*}
The calculations for the wedge-transition formulas \itref{it:wt} are analogous. \qed
\end{Cor}

\begin{Rem}
  \label{rem:suni}
  Expressions \itref{it:sd} and \itref{it:wt} here simplify further by
  noting that it follows immediately from the definition of \(S_{Q_\Ww}\) and its
  restrictions \(S_{Q_\Ww}^{\prec_\Ww/\succ_\Ww}\) that 
  \[
(S_{Q_{\Ww}}^{\succ_{\Ww}})^* =
S_{-Q_{\Ww}}^{\succ_{\Ww}} =
S_{Q_{\Ww'}}^{\succ_{\Ww}},
\numberthis
  \]
  where in the last equality the definition of \(Q_\Ww\) was used.
  This further implies 
  that these deformation maps are unitary, by writing
  \[
(S_{Q_{\Ww}}^{\succ_{\Ww}})^* S_{Q_{\Ww}}^{\succ_{\Ww}} 
= S_{-Q_{\Ww}}^{\succ_{\Ww}} S_{Q_{\Ww}}^{\succ_{\Ww}} 
= S_{-Q_{\Ww}+Q_{\Ww}}^{\succ_{\Ww}} = \Id,
\numberthis
  \]
  and analogously \(S_{Q_{\Ww}}^{\succ_{\Ww}} (S_{Q_{\Ww}}^{\succ_{\Ww}})^* =
  \Id.\) 
\end{Rem}

Let us note that we can further identify the scattering data for opposite wedges
by making use of the swapping property. For this consideration it is immaterial
whether we are in the deformed or undeformed model, so we will drop the
corresponding indices from the wave operators.
We start with an outgoing scattering state given by
\begin{align*}
  \Psi_n^+ 
  = \lim_{\T\to\infty} \BB_{1\T}(f_1) \ldots \BB_{n\T}(f_n)\Omega
  = \WO^+_{\Ww} (\BB_{1\T}(f_1) \Omega) \otimes \ldots \otimes
  (\BB_{n\T}(f_n)\Omega)
  \in \WO^+_{\Ww} \Gamma^{\succ_{\Ww}}(\HilbertSpace_1).
  \numberthis
\end{align*}
By an analogous argument as in \eqref{eq:normclucomm} and 
\eqref{eq:clucommdecay} we write
\begin{align*}
  \BB_{1\T}(f_1) \ldots \BB_{n\T}(f_n)\Omega &=
   \BB_{1\T}(f_1) \ldots \BB_{n-1\T}(f_{n-1}) \BB_{n\T}^\perp(f_n)\Omega
\\   &= \BB_{n\T}^\perp(f_n) \BB_{1\T}(f_1) \ldots \BB_{n-1\T}(f_{n-1}) \Omega
   + \text{(commutators)},
  \numberthis
\end{align*}
where the commutator terms are rapidly decreasing faster than any polynomial in
\(\T>0\).
Iterating this swapping argument we obtain
\begin{align*}
  \Psi_n^+ 
  = \lim_{\T\to\infty} \BB_{n\T}^\perp(f_n)\ldots  \BB_{1\T}^\perp(f_{1}) \Omega
  = \WO^+_{\Ww'} (\BB_{n\T}(f_n)\Omega) \otimes \ldots \otimes (\BB_{1\T}(f_1) \Omega)
  \in \WO^+_{\Ww'} \Gamma^{\succ_{\Ww'}}(\HilbertSpace_1).
  \numberthis
\end{align*}
In the last equality we already used that by swapping 
\(\BB_{k\T}^\perp(f_k)\Omega = \BB_{k\T}(f_k)\Omega\) and that
\begin{align*}
  \VS_{f_1} \succ_\Ww \VS_{f_2} \succ_\Ww \ldots
  \succ_\Ww \VS_{f_n} 
  &
  \Longleftrightarrow
  \VS_{f_n} \succ_{\Ww'} \VS_{f_{n-1}} \succ_{\Ww'} \ldots
  \succ_{\Ww'} \VS_{f_1} 
  \\ &
  \Longleftrightarrow
  \VS_{f_n} \prec_{\Ww} \VS_{f_{n-1}} \prec_{\Ww} \ldots
  \prec_{\Ww} \VS_{f_1} 
  \numberthis
\end{align*}
by the definition of the precursor relation and from the fact that \(\Ww' = - \Ww\) 
for any centered wedge \(\Ww\).
Let us therefore define \(Z : \FS^u(\HilbertSpace_1) \to
\FS^u(\HilbertSpace_1)\) by its action on \(n\)-particle states,
\[
  Z \Psi_1^1 \otimes \ldots \otimes  \Psi_1^n := 
  \Psi_1^n \otimes \ldots \otimes  \Psi_1^1, 
  \numberthis
\]
for any \(n \in \NaturalNum\) and \(\Psi_1^1, \ldots, \Psi_1^n \in
\HilbertSpace_1\) and we note that \(Z\) is a unitary and self-adjoint
involution. 
Summarizing these considerations we obtain:
\begin{Prop} For any centered wedge \(\Ww\) we have as subspaces of
  \(\Gamma^u(\HilbertSpace_1)\),
  \begin{enumerate}[(i)]
    \item 
      \( \FS^{\succ_{\Ww'}}(\HilbertSpace_1) =
      \FS^{\prec_{\Ww}}(\HilbertSpace_1)\),
    \item \(Z \FS^{\succ_\Ww}(\HilbertSpace_1)
        = \FS^{\succ_{\Ww'}}(\HilbertSpace_1) 
      \), and analogously
      \(
Z \FS^{\prec_\Ww}(\HilbertSpace_1)
        = \FS^{\succ_{\Ww}}(\HilbertSpace_1) 
      \).
  \end{enumerate}
\end{Prop}
\begin{Prop}
 \label{prop:wz}
 For any centered wedge \(\Ww\) the wave operators associated to complementary
 wedges can be identified by
 \begin{align*}
   \WO^+_\Ww &= \WO^+_{\Ww'} Z, &
   \WO^-_\Ww &= \WO^-_{\Ww'} Z. 
   \numberthis
 \end{align*}
\end{Prop}

This identification appears to take us somewhat away from the
localization properties appearing in the Haag-Ruelle construction.
Yet it shows that the scattering matrices for the two simplest choices
\(\Ww_\fin = \Ww_\ini\) and \(\Ww_\fin = \Ww_\ini'\) contain the same 
scattering theoretic data. These choices correspond to ones made in
\cite{GL07,BS08} for the analysis of two-particle scattering.  In these works
asymptotic two-particle states are constructed following more closely the
methods of Haag-Ruelle theory from local QFT. Hence a mixed localization is used,
where 
\[
  \Psi_2^+ = \lim_{\T \to \infty} \Bb_{1\T}^\perp(f_1) \Bb_{2\T}(f_2) \Omega
   = \lim_{\T \to \infty} \Bb_{2\T}(f_2) \Bb_{1\T}^\perp(f_1)  \Omega
  , \quad \VS_{f_2} \succ_\Ww
  \VS_{f_1}.
  \numberthis
\]
The independence of the operator order is reminiscent of Haag-Ruelle scattering
theory for bosonic local QFT. In the general wedge-local scattering theory it is
a special feature appearing for the case of two particles. It was already
remarked in \cite{MD17a} that this definition of two-particle scattering states
captures the same information as our analysis restricted to the level of
two-particle states.  This can be seen by swapping the corresponding sides,
\[
  \Psi_2^+ = \lim_{\T \to \infty} \Bb_{1\T}^\perp(f_1) \Bb_{2\T}^\perp(f_2) \Omega
   = \lim_{\T \to \infty} \Bb_{2\T}(f_2) \Bb_{1\T}(f_1)  \Omega.
   \numberthis
\]
Thereby we have not only illustrated the symmetry of \Cref{prop:wz} at the
two-particle level.
We also see that the compatibility of our analysis with earlier
calculations from \cite{GL07,BS08} is in fact a corollary of an intermediate step of the
proof of \Cref{prop:wz}.

\section{Asymptotic Completeness of BLS-Deformed wQFT} \label{sec:AC}

Given a wedge-local model \((\Alg, \alpha, \HilbertSpace, \Omega)\), we can now
proceed to our second main objective and study the completeness of asymptotic
states
\begin{align*}
  \HilbertSpace^\pm_{0,\Ww} := \WO^\pm_{0,\Ww} \FS^{\succ_\Ww/\prec_\Ww},\\
  \HilbertSpace^\pm_{Q,\Ww} := \WO^\pm_{Q,\Ww} \FS^{\succ_\Ww/\prec_\Ww}
  \label{eq:asymptsub}\numberthis
\end{align*}
in the common Hilbert space \(\HilbertSpace\) of the initial and BLS-deformed
wQFT model.
We will use a notion of asymptotic completeness for wQFT which directly
generalizes the corresponding standard definition from local QFT. 

\begin{Def} \label{def:ac}
  We say that a wave operator \(\WO^\pm_\Ww\) of a wQFT model for a centered localization wedge \(\Ww\)
  is \emph{asymptotically complete}\/ iff the subspace of
    velocity-ordered scattering states
    \begin{align*}
      \HilbertSpace^+_{\Ww} &:=
      \WO_\Ww^+ \FS^{\succ_\Ww}(\HilbertSpace_1), \; \text{or},
      &
      \HilbertSpace^-_{\Ww} &:=
      \WO_\Ww^- \FS^{\prec_\Ww}(\HilbertSpace_1),
      \label{eq:defac}\numberthis
    \end{align*}
    respectively, is dense in \(\HilbertSpace\).
    We say that a wQFT model 
    satisfies the property of \emph{ordered asymptotic
      completeness} (more precisely, \emph{asymptotic completeness with respect to
    velocity-ordered scattering states}) 
    iff both
    \(\WO^{+}_\Ww\) and \(\WO^-_\Ww\) are asymptotically complete for any wedge~\(\Ww\).
\end{Def}

The results of \Cref{sec:GL} give an explicit representation of the
deformed wave operators \(\WO^\pm_{Q,\Ww}\) 
in terms of the undeformed wave operators \(\WO^\pm_{0,\Ww}\).
This directly yields a general result regarding the stability of asymptotic completeness 
of wedge-local theories  under BLS-deformations.
\begin{Thm}\label{prop:ac}
  A wave operator of a deformed wQFT model~\(\WO_{Q,\Ww}^{\pm}\) is asymptotically complete if
  and only if the wave operator of the underlying ``undeformed'' model
  \(\WO_{0,\Ww}^{\pm}\) is asymptotically complete.
    \proof By
    \Cref{thm:defw}, we have
    \(\WO_{Q,\Ww}^{+} = \WO_{0,\Ww}^{+} S_{Q_\Ww}^{\succ_\Ww}\), and  
    \(S_{Q_\Ww}^{\succ_\Ww} \FS^{\succ_\Ww}(\HilbertSpace_1) = \FS^{\succ_\Ww}(\HilbertSpace_1)\) 
    by unitarity of \(S_{Q_\Ww}^{\succ_\Ww}\)  (see \Cref{rem:suni}). Hence
    \[
     \WO_{Q,\Ww}^{+} \FS^{\succ_\Ww}(\HilbertSpace_1)
     = \WO_{0,\Ww}^{+} S_{Q_\Ww}^{\succ_\Ww}\FS^{\succ_\Ww}(\HilbertSpace_1)
     = \WO_{0,\Ww}^{+} \FS^{\succ_\Ww}(\HilbertSpace_1).
     \numberthis
    \]
    Taking the closures yields the equivalence of ordered asymptotic
    completeness of deformed and initial model for outgoing states. The argument for the incoming case is
    analogous.
     \qed
  \end{Thm}

  In the following final section we will discuss the application of our results
  to GL-type models, which are constructed by applying BLS-deformations to a
  free field. For the scattering theoretic analysis of these models we have to
  work with only wedge-ordered states, as stated explicitly in \eqref{eq:defac}
  and dictated by the scope of the wedge-local scattering theory. 
  In GL-type models we will see this restriction to wedge-ordered states
  is inessential for the particle interpretation and still yields a dense set of
  scattering states. This may of course be expected on grounds of the
  bosonic statistics of the underlying free theory.
  Let us also note that there are wedge-local models, which do not satisfy
  ordered asymptotic completeness at the two-particle level. Examples of such
  models have been obtained by applying a von Neumann operator-algebraic free
  product construction to the free field \cite{LTU17}.

  \section{Application to Grosse-Lechner Models} \label{sec:AppGL}
  In the work of  Grosse and Lechner \cite{GL07} an interesting class of wedge-local models in
  any space-time dimension~\(d \geq 1+1\) is constructed in a wedge-local variant of the
  Wightman framework. A closely related class of models is obtained in the
  operator-algebraic framework by applying the BLS-deformation construction  to
  the standard scalar free field \cite{BS08,BLS10}. The scalar field has a
  canonical wedge-local description given by the von Neumann algebras
  \[
    \Alg^0(\Ww) := \{ W(f)  : f \in
      \SchwartzSpace(\RealNum^d,\RealNum),\;
      \support f \subset \Ww\}''
      \numberthis \label{eq:wedgealg}
  \] 
  generated by the Weyl operators~\(W(f) = \Ee^{\Ii \phi(f)}\), where
  \(\phi(f)\) is the standard free scalar Wightman field (see e.g. \cite{Dy17,Di11,IZ05}).
  Local algebras \(\Alg^0(\Reg)\) are defined analogously, requiring that
  \(\support f \subset \Reg\) for bounded open regions \(\Reg\) in Minkowski
  space-time.
  The algebras \(\Alg(\Ww)\) and \(\Alg(\Reg)\) act on the bosonic Fock space \(\HilbertSpace = \Gamma^b(\HilbertSpace_1)\)
  over the scalar one-particle space \(\HilbertSpace_1 =
  \LSpace^2(\RealNum^s)\), and we write \(\Omega_F \in \HilbertSpace\) for the Fock vacuum.
  The net is  covariant with respect to the standard second quantized scalar representation
  of the proper orthochronous Poincaré group.  
  For any non-vanishing warping matrix \(Q\in \RealNum^{d\times d}\) of the
  form \eqref{eq:warpedq} we will call the BLS-deformed theory~\((\Alg^Q, \alpha,
  \HilbertSpace, \Omega)\) a {\em Grosse-Lechner model}. Let us now apply our
  general results from Sections~\ref{sec:GL} and \ref{sec:AC} to evaluate the scattering theoretic content of these
  models and establish their ordered asymptotic completeness.

  It is a well-known fact that in the case of the free field also the
  conventional local Haag-Ruelle scattering theory applies (see e.g.\
    \cite{ArQFT99,Dy17}).  It yields wave operators \(\WO^\pm :
    \Gamma^b(\HilbertSpace_1) \to \HilbertSpace\), defined on the bosonic Fock
    space. For the free model we have
    \(\HilbertSpace=\Gamma^b(\HilbertSpace_1)\) and the wave operators are trivial in the sense that \(\WO^{+}
    = \WO^- = \Id\).\footnote{The argument for the two-dimensional case can be
    extracted from \cite{Le06t}~Lemma~6.1.1, applied for the (trivial) special case \(S=1\).} 
  This also directly implies asymptotic completeness
  of the free theory in the conventional sense with respect to the standard
  Haag-Ruelle scattering theory. That is, \(\WO^{\pm}
  \Gamma^b(\HilbertSpace_1)\) is dense in \(\HilbertSpace\) and in these cases we
  will say similarly that \(\WO^\pm\), respectively, are asymptotically
  complete.

  For the free field, the present velocity-ordered wedge-local formalism for scattering
  theory applies as well. The wedge-local wave operators \(\WO^\pm_{0,\Ww}\) can
  be efficiently determined from their local counterparts in cases for which the
  latter exist.  For this purpose we define  
  the  embeddings
  \[    
    I^{\succ_\Ww/\prec_\Ww} :
    \left\{
      \begin{aligned}\FS^{\succ_\Ww/\prec_\Ww}(\HilbertSpace_1) 
       &  \longrightarrow
    \FS^{b}(\HilbertSpace_1),
    \\
    \Psi_1^1 \otimes \ldots \otimes \Psi_1^n &\longmapsto
    a^*(\Psi_1^1) \ldots a^*(\Psi_1^n)\Omega_F = 
    \sqrt{n!}
    \cdot
    \Psi_1^1 \otimes_s \ldots \otimes_s \Psi_1^n ,
  \end{aligned}
  \right. 
  \label{eq:embed}
  \numberthis
  \]
  which map wedge-ordered \(n\)-particle vectors
   to the corresponding bosonic symmetrized tensor product.
 In terms of the norm on
   \(\FS^{b}(\HilbertSpace_1)\) we get 
   \[
     \norm{
       \sqrt{n!}
       \cdot
     \Psi_1^1 \otimes_s \ldots \otimes_s \Psi_1^n}^2
     =  \sum_{\pi \in \SymmetricGroup_n}
      \prod_{k=1}^n \langle \Psi_1^k, \Psi_1^{\pi(k)} \rangle
     = 
     \prod_{k=1}^n  \norm{\Psi_1^k}^2.
     \numberthis
   \]
  Here it is used that, as a consequence of ordered velocity supports, the one-particle
  states~\(\Psi_1^k \in \HilbertSpace_1\) (\(1\leq k \leq n\))
  are pairwise orthogonal. 
  Thus \( I^{\succ_\Ww/\prec_\Ww}\) are well defined  by linear continuous
  extension from \eqref{eq:embed} and yield isometries  on
  \(\FS^{\succ_\Ww/\prec_\Ww}(\HilbertSpace_1)\).   

\begin{Thm} \label{cor:acFree}
    \def\WedA{{\Ww_1}}
      \def\WedB{{\Ww_2}}
      The wedge-local wave operators of a local quantum field theory with
      isolated mass shell are well-defined and can be expressed using local
      Haag-Ruelle wave operators~\(\WO^\pm : \FS^b(\HilbertSpace_1) \to
      \HilbertSpace\) as
  \[
    \WO^{+}_\Ww = \WO^+ I^{\succ_\Ww}, \quad \WO^{-}_\Ww = \WO^- I^{\prec_\Ww}.
       \label{eq:waveLoc} \numberthis
    \]
      Further, the S-matrix and wedge-transition maps in the
    case of a  local quantum field theory are given by
    \begin{align*}
      S_{\fin\,\ini}^{\WedA\WedB} &= (I^{\succ_\WedA})^* S_{\fin\, \ini} I^{\prec_{\WedB}}, 
                                   \; \text{and},
                                   \numberthis\label{eq:S1loc}
\\
S_{\fin\,\fin}^{\WedA\WedB} &= (I^{\succ_\WedA})^* I^{\succ_{\WedB}}, \quad
                            &S_{\ini\,\ini}^{\WedA\WedB} &= (I^{\prec_\WedA})^* I^{\prec_{\WedB}},
      \label{eq:Sloc} \numberthis
    \end{align*}
    where \( S_{\fin\, \ini} = (\WO^+)^*\WO^-\) denotes the usual scattering
    matrix from local Haag-Ruelle theory.
    \proof It is sufficient to establish \eqref{eq:waveLoc}, from which the other
    statements follow. Here we note that for the wedge-transition matrix
    identities \eqref{eq:Sloc} also the Fock structure
    \((\WO^+)^*\WO^+ = \Id = (\WO^-)^*\WO^-\) of scattering states in the local
    Haag-Ruelle theory is used.

    We consider the case of the outgoing wave operator. We may restrict to
    one-particle vectors of the form \(\Psi_k = \Bb_{k\T}(f_k)\Omega\) defined
    in terms of local operators \(A_k
      \in \Alg(\Reg)\) and regular Klein-Gordon solutions \(f_k\) for \(1\leq k
    \leq n\). In a local
      QFT model such vectors yield a dense subset of the one-particle space
  by standard arguments.
      We further assume that the one-particle states (and similarly the wave
      packets) are velocity ordered \(\Psi_1 \succ_\Ww \Psi_2 \succ_\Ww \ldots \succ_\Ww \Psi_n\).
      Lastly, these states trivially satisfy the swapping property with
      \(A^\perp_k := A_k\) for some overlapping wedge-regions \(\Ww \supset
      \Reg\), \(\Ww^\perp \supset \Reg\).
      Thus, we have on one hand by definition of the wedge-local wave operators that
      \begin{align*}
       \lim_{\T\to\infty} \Bb_{1\T}(f_1) \ldots \Bb_{n\T}(f_n)\Omega
       &=
      \WO^{+}_\Ww \Psi_1 \otimes \Psi_2 \otimes \ldots \otimes \Psi_n.
     \label{eq:freemol}\numberthis
    \end{align*}
    On the other hand, we obtain from standard local Haag-Ruelle theory that
    \begin{align*}
       \lim_{\T\to\infty} \Bb_{1\T}(f_1) \ldots \Bb_{n\T}(f_n)\Omega
       &= \WO^{+} a^*(\Psi_1) \ldots a^*(\Psi_n) \Omega_F
       = \WO^{+} I^{\prec_{\Ww}}
\Psi_1 \otimes \Psi_2 \otimes \ldots \otimes \Psi_n.
\numberthis
    \end{align*}
    Equating we obtain \eqref{eq:waveLoc} on a total subset and thus by
    continuity on \(\FS^{\succ_\Ww}(\HilbertSpace_1)\).
    The argument for the incoming case is analogous.
   \qed
  \end{Thm}

Using that the spectrum of the momentum operators is
  Lebesgue absolutely continuous on the one-particle space,
  \(I^{\succ_\Ww/\prec_\Ww}\) are in fact surjective. Such continuity 
  clearly holds in the free example \eqref{eq:wedgealg}. In a general context
  it is known to follow from locality and the spectrum condition
  \cite[Prop.\ 2.2]{BF82}, or, from Poincaré covariance \cite{Mai68}.
  \begin{Prop}
    For a local QFT model the following statements are equivalent:
    \begin{enumerate}[(i)]
      \item \(\WO^+\) is asymptotically complete,
      \item \(\WO^+_\Ww\) is asymptotically complete for one wedge \(\Ww\),
      \item \(\WO^+_\Ww\) are asymptotically complete for all wedges \(\Ww\),
    \end{enumerate}
    and analogously for the completeness of incoming wave operators.
    \proof The equivalence of (i) and (ii) for any wedge \(\Ww\) follows from
    \(I^{\succ_\Ww/\prec_\Ww}\) being a surjective isometry and
    \eqref{eq:waveLoc}. As \(\Ww\) was arbitrary, this implies (iii). \qed
  \end{Prop}

  \begin{Cor}
      The wedge-local wave operators of the free scalar field are given by
  \[
    \WO^{+}_{0,\Ww} = I^{\succ_\Ww}, \quad \WO^{-}_{0,\Ww} = I^{\prec_\Ww},
       \label{eq:waveFree} \numberthis
    \]
    In particular, the free scalar field satisfies the property of ordered
    asymptotic completeness.
  \end{Cor}

  \begin{Thm}
    The Grosse-Lechner models are asymptotically complete with respect to
    velocity-ordered scattering states for any warping matrix \(Q\) (as defined in
    \eqref{eq:warpedq}).
    \proof We consider a GL-model for any fixed warping matrix \(Q\).
    \Cref{cor:acFree} shows that the wedge-local wave operators~\(\WO^{\pm}_{0,\Ww}\) 
    of the scalar free field are asymptotically
    complete for any wedge~\(\Ww\). 
    By \Cref{prop:ac} we obtain ordered asymptotic completeness of
    the wave operator~\(\WO^{\pm}_{Q,\Ww}\) for any \(\Ww\) and, thereby, ordered 
    asymptotic completeness of the considered GL-model.
    \qed
  \end{Thm}

  To conclude these investigations, let us note the explicit \(n\)-particle
  scattering matrix for the case of the Grosse-Lechner models.
  We will focus on the two cases of equal or opposite initial and final wedges,
which correspond to the earlier investigations from \cite{GL07,BS08}.
The case \(\Ww_\fin = \Ww'_\ini\) can perhaps be regarded as slightly more
natural due to the coincidence of \(\FS^{\prec_\Ww}(\HilbertSpace_1)=
\FS^{\succ_{\Ww'}}(\HilbertSpace_1)\) as subspaces of
\(\FS^u(\HilbertSpace_1)\), which is a consequence of the equivalence of  the
ordering relations \(\succ_{\Ww'}\) and \(\prec_{\Ww}\). 
  \begin{Prop} Let \(\Ww\) be a fixed initial wedge. Then for the case of \(\Ww_\fin :=
      \Ww'\) we obtain the Grosse-Lechner \(S\)-matrix as unitary operator on
      \(\FS^{\prec_\Ww}(\HilbertSpace_1)\) given by
  \begin{align*}
    S_{Q,\fin\, \ini}^{\Ww'\Ww} =  
    (S_{Q_{\Ww'}}^{\succ_{\Ww'}})^* S_{Q_\Ww}^{\prec_\Ww} 
    = (S_{Q_\Ww}^{\prec_{\Ww}})^2 = S_{2Q_\Ww}^{\prec_{\Ww}}.
    \label{eq:calcs}\numberthis
  \end{align*}
  Further we have for any ordered \(n\)-particle state \(\Psi_n = \Psi_1^1 \otimes \ldots
  \otimes \Psi_1^n \in \FS^{\prec_\Ww}(\HilbertSpace_1)\) that
\[
S^{\prec_\Ww}_{2Q_\Ww} \Psi_n = \prod_{1\leq i < j \leq n} \Ee^{2\Ii P_i\cdot Q_\Ww P_j}
\Psi_n.
\numberthis
\]
  In particular, the Grosse-Lechner \(S\)-matrix is non-trivial and factorizing.
\end{Prop}
This follows directly from \Cref{cor:scattdef} and \Cref{cor:acFree}, together
with the triviality of the scattering matrix of the free field.  Further we also used
that \(\succ_{\Ww'}\) and \(\prec_{\Ww}\) are equivalent, giving \( (S_{Q_{\Ww'}}^{\succ_{\Ww'}})^*  =  S_{-Q_{\Ww'}}^{\prec_{\Ww}}
  =  S_{Q_{\Ww}}^{\prec_{\Ww}} \) by \Cref{rem:suni} and the definition of
  \(Q_\Ww\).
For the case of equal initial and final
wedges we obtain similarly
\begin{align*}
    S_{Q,\fin\, \ini}^{\Ww\Ww} =  
    (S_{Q_{\Ww}}^{\succ_{\Ww}})^* (I^{\succ_\Ww})^* I^{\prec_\Ww} S_{Q_\Ww}^{\prec_\Ww} 
    = S_{-Q_{\Ww}}^{\succ_{\Ww}} Z S_{Q_\Ww}^{\prec_\Ww} 
    = Z (S_{Q_\Ww}^{\prec_{\Ww}})^2 = Z S_{2Q_\Ww}^{\prec_{\Ww}}.
    \label{eq:calcsww}\numberthis
  \end{align*}
  We note that the same result is obtained when using \Cref{prop:wz} and
  \eqref{eq:calcs}.

\section{Conclusions}
In this paper we showed stability of ordered asymptotic completeness under BLS-deformations in wedge-local QFT.
We concluded that the Grosse-Lechner model is interacting and asymptotically complete in any spacetime dimension.
We also showed that this model has a factorizing $S$-matrix, which is an unusual feature in higher dimensions.

Although we focussed on models obtained by BLS deformations, our approach should also apply
to wedge-local factorizing models  in two dimensions. Even in the cases in which strict locality is still open, such as the 
non-linear sigma models \cite{AL17}, our strategy may give asymptotic completeness and factorization of the $S$-matrix.
We leave  detailed analysis of these problems to future investigations.

Another natural direction  is a generalization of our results  to theories of massless particles.
The BLS-deformations remain valid for such theories, but so far neither interaction nor asymptotic
completeness have been studied for $d>1+1$. Such an investigation would require a massless
variant of wedge-local scattering theory from \cite{MD17a}. It may be difficult to develop such a theory 
at the same level of generality as its local counterpart \cite{Bu77}, since the decay of correlations in massless
wedge-local models may be very slow.  Also the energy bounds \cite{Bu90}, which
simplify more recent constructions of massless scattering states~\cite{DH15,AD15},
are not available for wedge-local theories. But under some natural assumptions on the decay of correlations  
massless wedge-local scattering theory appears to be within reach.  It should apply, in particular, to the massless Grosse-Lechner model.

We mention as an aside, that such a scattering theory could also help to understand recent
computations of  infraparticle scattering amplitudes in four-dimensional
string-local models \cite{GRT21}. 
An apparent breakdown of unitarity of scattering amplitudes was found in this reference after
adapting a formula from a two-dimensional context \cite{DM21}.
This problem may have its roots in our limited understanding
of collisions of massless Wigner particles  in string- and wedge-local theories.

\printbibliography
\end{document}